\documentclass[aps,pra,showpacs,twoside,twocolumn,10pt]{revtex4-2}
\usepackage[colorlinks=true, citecolor=blue, urlcolor=blue, linkcolor = blue ]{hyperref}

\usepackage{graphicx}
\usepackage{bm}
\usepackage{setspace}
\usepackage{amsmath,amsfonts}
\usepackage{amsthm}

\usepackage{mathrsfs}
\usepackage{subfigure}
\usepackage{xcolor}

\usepackage{comment}
\usepackage{braket}
\theoremstyle{plain}
\usepackage{color}
\usepackage{amssymb}
\usepackage{amsthm}
\usepackage{amsfonts}
\usepackage{float}
\usepackage{tabularx}
\usepackage{graphicx}
\usepackage[export]{adjustbox}
\usepackage{mathtools}
\usepackage{esvect}
\usepackage{wrapfig}
\usepackage{amsthm}
\usepackage{verbatim}
\usepackage{bbm}
\usepackage[normalem]{ulem}

\usepackage{enumitem}
\usepackage{fmtcount}
\usepackage{booktabs}
\usepackage{csquotes}
\usepackage{epsfig}

\usepackage{tabularx}
\usepackage{graphicx}
\usepackage{amsmath}
\usepackage{braket}
\usepackage{latexsym}
\usepackage{bm}
\usepackage{graphics,epstopdf}
\usepackage{enumitem}
\usepackage{fmtcount}
\usepackage{booktabs}
\usepackage{csquotes}
\usepackage{epsfig}
\usepackage{hyperref}
\theoremstyle{plain}

\def\bea{\begin{eqnarray}}
\def\eea{\end{eqnarray}}
\def\ba{\begin{array}}
\def\ea{\end{array}}

\def\beq{\begin{equation}}
\def\eeq{\end{equation}}

\usepackage[normalem]{ulem}
\usepackage{float}
\usepackage{graphicx}  
\usepackage{dcolumn}          
\usepackage{amssymb}
\usepackage{appendix}
\usepackage{physics}   
\usepackage{mathtools}
\usepackage{esvect}
\usepackage{wrapfig}
\usepackage{amsthm}
\usepackage{verbatim}
\usepackage{bbm}

\usepackage[mathscr]{euscript}
\def\Tr{\operatorname{Tr}}

\def\({\left(}
\def\){\right)}
\def\[{\left[}
\def\]{\right]}

\newtheorem{theorem}{Theorem}

\newtheorem{Observation}{Observation}

\begin{document}

\title{Fluctuation in energy extraction from quantum batteries: How open should the system be to control it?}






\author{Anindita Sarkar}
\email{aninditasarkar@hri.res.in}

\author{Paranjoy Chaki} 
\email{paranjoychaki@hri.res.in}

\author{Priya Ghosh} 
\email{priyaghosh@hri.res.in}

\author{Ujjwal Sen}
\email{ujjwal@hri.res.in}

\affiliation{Harish-Chandra Research Institute,  A CI of Homi Bhabha National Institute, Chhatnag Road, Jhunsi, Prayagraj  211 019, India}

\begin{abstract}
We ask whether there exists a relation between controllability of the fluctuations in extractable energy of a quantum battery and (a) how open an arbitrary but fixed battery system is and (b) how large the battery is. We examine  
three classes of quantum processes for the energy extraction: unitary operations, completely positive trace-preserving 
(CPTP) 
maps, and arbitrary quantum maps, including physically realizable 
non-CPTP 
maps.
We show that all three process classes yield the same average extractable energy from a fixed quantum battery. Moreover, open systems are better at controlling fluctuations in fixed quantum batteries:  while random unitary operations result in nonzero fluctuation in the extractable energy,  the remaining two classes 
lead to vanishing fluctuations in extractable energy. Furthermore, when the auxiliary system used to implement 
the non-unitary physically realizable maps
is restricted up to a dimension $n$, fluctuation in extractable energy scales as $1/n$ for CPTP maps, outperforming the $\ln{n}/n$ scaling observed for general quantum maps.
Even within open dynamics, therefore, energy extraction via random CPTP maps exhibits greater resilience to fluctuation compared to processes based on arbitrary quantum maps.
We subsequently obtain that  fluctuations in extractable energy scale as the inverse of the battery's dimension for all three process classes.
Unitary maps, therefore, perform - in the sense of as low fluctuation as possible - equally well as more resource-intensive open maps, provided we have access to large quantum batteries.
The results underscore a fundamental trade-off between a battery’s performance and the resource cost of implementing the extraction processes.
\end{abstract}
\maketitle

\section{Introduction}
The conceptualization and implementation of quantum thermal 
devices~\cite{Bender_2000,PhysRevE.76.031105,PhysRevLett.105.130401,doi:10.1142/S1230161214400022,PhysRevLett.116.200601,Vinjanampathy01102016,Bhattacharjee2021,Myers_2022,potts2024quantumthermodynamics,CANGEMI20241} has gained growing importance in present-day quantum technologies. Due to the inherent fluctuations in thermodynamic quantities at the quantum scale, a deep understanding of these fluctuations is crucial for assessing the robustness and performance of thermal machines operating in the quantum regime.
Traditionally, the role of fluctuations in quantum thermodynamics has been described mainly through the formulation of fluctuation-dissipation relations~\cite{RevModPhys.81.1665,PhysRevResearch.1.033156,PhysRevResearch.3.023252}, fluctuation Theorems~\cite{PhysRevLett.92.230602,RevModPhys.83.771,PRXQuantum.1.010309}, thermodynamic uncertainty relations~\cite{PhysRevE.103.012111,PhysRevLett.126.010602,Menczel_2021,Razzoli2024} etc. Recently, there has been a lot of interest in studying the effect of fluctuations occurring in quantum devices~\cite{PhysRevE.103.012133,PhysRevLett.127.190603,Gerry_2022,PhysRevE.108.014137}, such as quantum thermal transistors~\cite{das2024fluctuationsoptimalcontrolfloquet}, quantum heat engines~\cite{Campisi_2014,Campisi_2015,PhysRevA.92.033854,PhysRevE.94.032116,PhysRevE.99.022104,PhysRevLett.123.080602,PhysRevA.101.010101,PhysRevResearch.2.032062,_obejko_2020,PhysRevE.103.032130,Bouton2021,PhysRevE.103.L060103,PhysRevA.104.012217,PhysRevResearch.3.L032041,denzler2021nonequilibriumfluctuationsquantumheat,CANGEMI20241,Razzoli2024,PhysRevE.106.014114,PhysRevE.106.014143,PhysRevE.106.024123,10.21468/SciPostPhys.12.5.168,alam2022twostrokequantummeasurementheat,li2022performancequantumheatengines,xiao2022finitetimequantumottoengine,PhysRevA.107.L040202,PhysRevA.108.032203,PhysRevE.103.012111,Chesi_2023,PhysRevA.104.012217,PhysRevE.108.044114,PhysRevE.106.024123,xiao2023thermodynamicsfluctuationsquantumheat,sarmah2024efficiencyfluctuationsheatengine}, quantum refrigerators~\cite{Campisi_2014,CANGEMI20241,PhysRevE.97.052145,Manzano2018,Jiao_2021,PhysRevApplied.19.034023}, quantum batteries~\cite{Friis_2018,random-battery-1,Crescente_2020,PhysRevLett.125.040601,10.21468/SciPostPhys.10.3.076,PhysRevLett.127.028901,e23111455,e24060820,random-battery-2,PhysRevE.109.014131,PhysRevA.110.022226,zahia2024entanglementdrivenenergyexchangetwoqubit}.

Quantum batteries~\cite{battery-first,b1,Bhattacharjee2021,RevModPhys.96.031001} are one of the quantum thermal devices that act as miniature energy storage devices. Over the past decade, research in quantum batteries has made significant strides in quantum technology; with advances in both charging~\cite{Bhattacharjee2021,C1,C6,PhysRevA.104.032207,saha2023harnessingenergyextractedheat,Downing_2024,PhysRevResearch.6.023136,PhysRevA.110.022425,puri2024floquetdrivenlongrangeinteractions,Shastri2025,ghosh2025constructiveimpactwannierstarkfield} and discharging~\cite{Allahverdyan_2004,PhysRevE.90.012127,PhysRevLett.123.190601,Gumberidze2019,e21080771,PhysRevA.104.L030402,PhysRevA.104.042209,chaki2023auxiliaryassistedstochasticenergyextraction,PhysRevLett.132.240401,chaki2024positivenonpositivemeasurementsenergy,one,chaki2024universalcompleteextractionenergyinvariant} protocols, as well as capacity analysis~\cite{PhysRevLett.131.030402}.
Moreover, studies on quantum batteries have extended to many-body spin chains~\cite{PhysRevA.101.032115,PhysRevA.105.022628}, non-Hermitian systems~\cite{PhysRevA.109.042207}, ultracold atoms~\cite{PhysRevA.106.022618},  noisy scenario~\cite{sen2023noisyquantumbatteries}, nonlinear system~\cite{bhattacharyya2024nonlinearityassistedadvantagechargersupportedopen}, and so on.
Experimentally, quantum batteries have been realized in NMR spin systems~\cite{PhysRevA.106.042601}, superconducting qubits~\cite{Hu_2022}, and organic microcavities~\cite{doi:10.1126/sciadv.abk3160}.

In this work, we investigate the robustness of the performance of a quantum battery under three types of energy extraction processes: random unitary operations, random completely positive trace-preserving (CPTP) maps, and arbitrary random quantum maps, including noncompletely positive trace-preserving (NCPTP) maps. In particular, we study the robustness of a quantum battery’s performance by analyzing two key aspects for each process—(i) the average amount of energy that can be extracted and (ii) the fluctuation in  extractable energy.

We find that the average extractable energy from a quantum battery remains the same, regardless of whether the extraction process involves random unitary operations, random CPTP maps, or arbitrary random quantum maps (including NCPTP maps). However, the fluctuation in  extractable energy differ for different energy extraction processes. Specifically, random unitary operations lead to finite fluctuation in extractable energy, except when the battery is initially in a maximally mixed state. In contrast, when random CPTP or general quantum maps (including NCPTP) are used as energy extraction processes, the fluctuations in extractable energy vanish.
Note that implementation of CPTP and NCPTP maps necessitate an auxiliary system and thus require significantly more resources compared to unitary operations. The auxiliary system should be initially in a product state with the battery for CPTP maps, whereas entanglement between the battery and the auxiliary is necessary to realize NCPTP maps. Hence, achieving lesser fluctuations in extractable energy demands more quantum resources, highlighting a trade-off between the resource cost and the efficiency of the quantum processes used for energy extraction from the battery.

Furthermore, when the dimension of the auxiliary system used to implement non-unitary, physically realizable maps is limited to a maximum of $n$, the behavior of fluctuations in extractable energy reveals a clear distinction between CPTP maps and general quantum maps. Specifically, for CPTP maps, the fluctuation in extractable energy decreases proportionally to $1/n$, indicating a faster decay as the auxiliary system's dimension increases. In contrast, for general quantum maps (including non-CPTP but physically realizable ones), the fluctuation follows a slower decay, scaling as $\ln n / n$.
This means that, under the constraint on the auxiliary system's dimension, CPTP maps offer a more favorable reduction of fluctuation in extractable energy, making them more efficient and stable for energy extraction tasks as the size of the auxiliary increases.

We now compare the fluctuations in extractable energy from a quantum battery under three different energy extraction processes: (i) random unitary operations, (ii) random CPTP maps constructed using a fixed-dimensional auxiliary system, and (iii) arbitrary quantum maps—including NCPTP maps—implemented with auxiliary systems of the same fixed dimension, wherever such an auxiliary is required.
Our findings reveal that random CPTP maps constructed with a fixed-dimensional auxiliary consistently produce more fluctuation in extractable energy compared to those arising from arbitrary quantum maps (including NCPTP) constructed with the same auxiliary dimension. Furthermore, the fluctuation in extractable energy due to random unitary operations can be either greater or smaller than those from random CPTP or general random quantum maps. However, in the limit of large battery dimension, random unitary operations always lead to higher fluctuation in extractable energy as compared to random CPTP maps implemented with an auxiliary system of fixed dimension.

Interestingly, for a fixed-dimensional quantum battery, when random CPTP and arbitrary quantum maps (including NCPTP) implemented using an auxiliary system of the same dimension $d_A$ are used as energy extraction processes, the fluctuations in  extractable energy from the battery scale as $1/d_A^2$ and $1/d_A$, respectively.
Therefore, in the case of open evolution, implemented via CPTP or more general quantum maps—including physically realizable non-CPTP maps—CPTP maps surprisingly outperform the broader class, despite non-CPTP operations requiring more resources. This highlights a trade-off between quantum battery performance and the resource cost of extraction protocols.

On the other hand, fluctuations in extractable energy from a quantum battery with a dimension $d_B$ scale as $1/{d_B}$, when the energy extraction process is random unitary or random CPTP or an arbitrary quantum map – including NCPTP ones. Therefore, for large-dimensional quantum batteries, unitary maps perform just as well as more resource-intensive open maps in terms of minimizing fluctuations, which increases the preferability of using a unitary map.

The remaining part of the paper is structured as follows. In Sec.~\ref{not}, we introduce the notations that will be used throughout the paper.  In Sec.~\ref{sec2}, we explore about the key ideas and concepts, including quantum maps and quantum battery. In Sec.~\ref{sec3}, 
we analyze the average energy that can be extracted using different types of quantum processes: random unitary maps, random CPTP maps and arbitrary random quantum maps. We find that all three processes yield the same average extractable energy.
In Sec.~\ref{sec4}, we study the behavior of the fluctuation in extractable energy generated by each process. We observe that although a finite amount of fluctuation in extractable energy persists in the case of unitary maps, the fluctuations in extractable energy are zero when CPTP maps and arbitrary quantum maps (including NCPTP maps) are used to extract energy from a given quantum battery. In the same section,  we further determine the scaling of the finite-term averaged fluctuations in extractable energy in the open system dynamics with respect to a fixed auxiliary dimension $n$, where the auxiliary dimensions over which the averaging is carried out are restricted to values up to $n$. Fluctuation in extractable energy arising from CPTP maps shows a smaller scaling with $n$, compared to that produced by arbitrary quantum maps. In Sec.~\ref{sec5}, we demonstrate the different possible scenarios based on the magnitude wise ordering of the fluctuations in extractable energy resulting from all three kinds of quantum processes in the moderate and high $d_B$ regions. Furthermore, we inspect the scaling of the fluctuations in extractable energy arising from CPTP maps and arbitrary quantum maps, with respect to the auxiliary dimension $d_A$ in the asymptotic regime. 
 In Sec.~\ref{sec6} we examine the scaling of fluctuations in extractable energy for all three kinds of processes with the battery dimension $d_B$ in the asymptotic limit. Finally, we conclude in Sec.~\ref{sec7}.

\section{Notations} \label{not}
In this section, we introduce the notations for key physical quantities used throughout the paper. The space of linear operators on a Hilbert space $\mathcal{H}$ is denoted by $\mathcal{L}(\mathcal{H})$. The subset of $\mathcal{L}(\mathcal{H})$ consisting of all positive semi-definite operators with unit trace is denoted by $\mathcal{D}(\mathcal{H})$. Any quantum state or density operator defined on $\mathcal{H}$ belongs to $\mathcal{D}(\mathcal{H})$. The notation $\mathcal{H}^d$ denotes a Hilbert space with dimension $d$.
The group of unitary operators and and the identity operator on a $d$-dimensional Hilbert space are represented by $\mathcal{U}(d)$ and $\mathbb{I}_d$, respectively. 
The adjoint of any operator $A$ is denoted as $A^\dagger$.
The Haar average of a quantity K is indicated by $\overline{K}$. Trace of any operator, say $``\cdot"$, is represented by $\Tr{(\cdot)}$, while partial trace of any operator, $``\cdot"$, with respect to a system $E$ is represented as $\Tr_E{(\cdot)}$.

\section{Setting the stage} 
\label{sec2}
In this paper, we analyze qualitatively as well as quantitatively the average extractable energy and fluctuations in extractable energy from a quantum battery under various classes of quantum dynamics. This section provides the necessary background for our work. Subsection~\ref{qma} offers a discussion of quantum dynamics, followed by a concise overview of fundamental concepts related to quantum batteries in subsection~\ref{qb_b1}.

\subsection{Quantum dynamics}
\label{qma}
Let $\mathcal{B}(\Lambda)$ denote the set of all quantum maps. A quantum map $\Lambda \in \mathcal{B}(\Lambda)$ transforms a quantum state $\rho$ defined on a Hilbert space $\mathcal{D}(\mathcal{H})$
into another quantum state $\rho'$ acting on a (same or different) Hilbert space $\mathcal{D}(\mathcal{H}')$
, i.e., $\rho' \coloneqq \Lambda (\rho)$. 
In general, the quantum dynamics acting on any quantum state $\rho$ can be classified into three categories: $(i)$ unitary map, $(ii)$ completely positive trace-preserving (CPTP) map and $(iii)$ noncompletely positive trace-preserving (NCPTP) map. 
Throughout the paper, we use the abbreviations ``CPTP" for completely positive trace-preserving maps and ``NCPTP" for noncompletely positive trace-preserving maps. All of these three types of quantum maps are described below in detail. 

\begin{itemize}
    \item \textbf{Unitary map:}
A quantum map $\Lambda_{\mathrm{U}} : \mathcal{D}(\mathcal{H}^d) \rightarrow \mathcal{D}(\mathcal{H}^d)$ is called a unitary map if it transforms a quantum state $\rho \in \mathcal{D}(\mathcal{H}^d)$ into another state $\rho' \in \mathcal{D}(\mathcal{H}^d)$ such that
$\rho' \coloneqq \Lambda_{\mathrm{U}}(\rho) = U \rho U^\dagger$,
where $U$ is a unitary operator satisfying $U U^\dagger = U^\dagger U = \mathbb{I}_d$.

\item\textbf{CPTP map:} 
A map $\Lambda_\mathrm{CP}$ is called completely positive if it not only preserves the positive semi-definiteness of all positive semi-definite operators but also ensures that, for any arbitrary integer $d’ \geq 2$, the extended map $\Lambda_\mathrm{CP} \otimes \mathbbm{I}_{d’}$ maps all positive semi-definite operators on the extended Hilbert space to positive semi-definite operators.
Additionally, if a completely positive map preserves the trace of every operator it acts upon, it is referred to as a completely positive trace-preserving (CPTP) map. 
The action of any CPTP quantum map on a quantum state $\rho \in \mathcal{D}(\mathcal{H}^d)$ can be represented in the Kraus operator formalism as
$\Lambda_{\mathrm{CPTP}}: \rho \rightarrow \Lambda_{\mathrm{CPTP}}(\rho) \coloneqq \sum_i K_i \rho K_i^\dagger$,
where the Kraus operators $\{K_i\}$ satisfy the completeness relation $\sum_i K_i^\dagger K_i = \mathbbm{I}_{d}$.
If a CPTP map admits only a single Kraus operator, then the map automatically reduces to a unitary quantum map. Hence, the set of unitary quantum maps is a strict subset of CPTP quantum maps.
\vspace{-3mm}
\\
\item \textbf{NCPTP map:}
Beyond CPTP maps, there exists a broader class of positive maps known as noncompletely positive trace-preserving (NCPTP) maps. Since a quantum map admits a Kraus operator representation if and only if it is a CPTP map~\cite{10.1093/acprof:oso/9780199213900.001.0001}, NCPTP maps cannot be expressed in the Kraus operator formalism.
\end{itemize}

\begin{figure}
\includegraphics[scale=0.15]{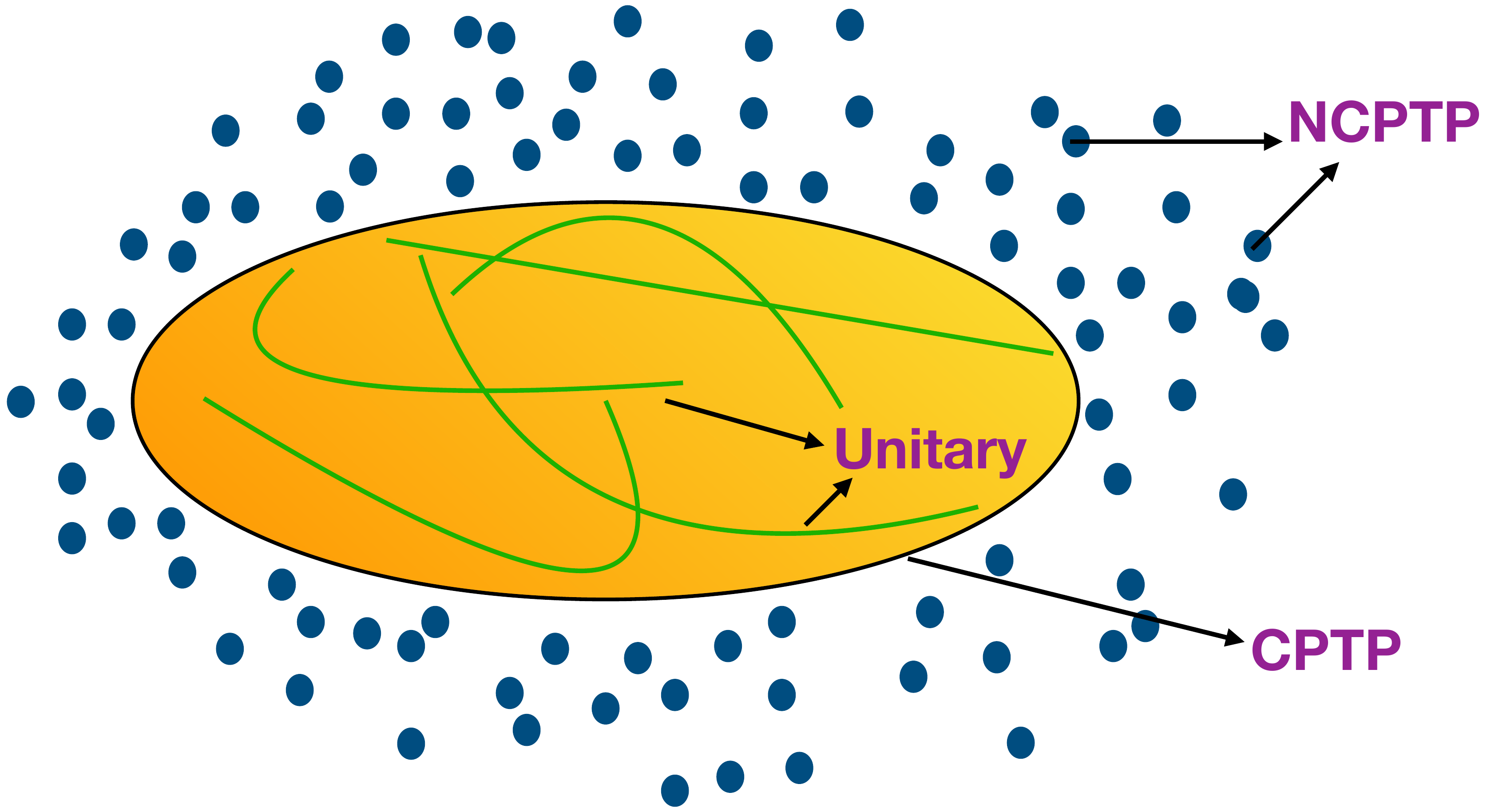}
\caption{\textbf{Schematic representation of all quantum dynamics.} 
The set of unitary dynamics, shown as green lines, forms a subset within the yellow oval-shaped region representing CPTP maps. Blue dots outside the yellow oval-shaped region denote NCPTP quantum maps. The black boundary around the yellow oval-shaped region highlights that the sets of CPTP and NCPTP maps are disjoint subsets. Together, they encompass the full set of quantum dynamics.}
\label{schem_2}
\end{figure}

The evolution of a closed quantum system is always governed by unitary dynamics. In contrast, open quantum systems — where the system interacts with an auxiliary system — evolve by either CPTP or NCPTP maps. Mathematically, the evolution of an open system is described by
$\rho^\prime_S = \Tr_A(U_{SA} \rho_{SA} U_{SA}^\dagger)$,
where $\rho^\prime_S$ denotes the final state of the system after interacting with the auxiliary through a global unitary operator $U_{SA}$, and $\rho_{SA}$ represents the initial joint state of the system and the auxiliary, denoted by $S$ and $A$ respectively.
Open system dynamics can be classified as CPTP or NCPTP, depending on the initial joint state of the system and the auxiliary. 
An initial product joint system-auxiliary state, that is, $\rho_{SA} = \rho_S \otimes \rho_A$, is sufficient for realizing CPTP dynamics~\cite{Nielsen_Chuang_2010}.
On the other hand, initial entanglement between the system and the auxiliary must be present to realize NCPTP dynamics.

Quantum dynamics can be hierarchically categorized into three subsets: a set of unitary maps, a set of CPTP quantum maps, and a set of NCPTP quantum maps. In this regard, the set of unitary dynamics forms a subset of the set of CPTP maps, which in turn is a subset of all quantum maps. Meanwhile, NCPTP maps represent a separate subset of all quantum dynamics that is disjoint from the set of CPTP maps. This classification is illustrated schematically in Fig.~\ref{schem_2}. In Fig.~\ref{schem_2}, green lines represent unitary dynamics as a subset of the yellow oval-shaped region denoting CPTP quantum maps. The blue dots outside the yellow oval-shaped region indicate NCPTP maps, while the black boundary around the yellow oval-shaped region implies that the sets of CPTP and NCPTP maps are disjoint subsets of all quantum maps.

\subsection{Quantum Battery}\label{qb_b1}
In this subsection, we move on to the discussion of quantum batteries. 
In general, a quantum battery~\cite{battery-first} is characterized by an initial quantum state and an associated Hamiltonian. Let the battery be initially in a state $\rho_B$, acting on a Hilbert space $\mathcal{H}_B$ of dimension $d_B$, with Hamiltonian $H_B$. Energy extraction from a quantum battery involves applying a quantum map that transforms the initial state of the battery to a final state $\rho^\prime_B$ such that the energy of the battery decreases, i.e., $\Tr(\rho^\prime_B H_B) \leq \Tr(\rho_B H_B)$, where $\Tr(\rho_B H_B)$ and $\Tr(\rho’_B H_B)$ denote the initial and final energies of the battery, respectively. The amount of energy extracted from the battery, denoted by $\mathcal{E}$, is defined as
$\mathcal{E} \coloneqq \Tr(\rho_B H_B) - \Tr(\rho’_B H_B)$.
The concept of a quantum battery is closely tied to the ideas of ergotropy and passive states. The maximum energy extractable from a quantum battery via unitary operations is known as ergotropy~\cite{battery-first,Allahverdyan_2004}, and is defined as
\begin{align*}
   \mathcal{E}_{\mathrm{max}} &\coloneqq \Tr(\rho_B H_B) - \min_{U_B} \Tr(U_B \rho_B U_B^\dagger H_B), \\
   &= \Tr(\rho_B H_B) - \Tr(\sigma_p H_B), 
\end{align*}
where $U_B$ is any unitary operator that acts on the battery. $\sigma_p$ denotes the state from which no energy can be extracted by any unitary processes from the battery, i.e., 
$\Tr(\sigma_p H_B) \leq \Tr(U_B \rho_B U_B^\dagger H_B) \hspace{2mm} \forall \hspace{1mm} U_B$. Such states are known as~\emph{passive states}~\cite{Lenard1978,battery-first}.

With the necessary groundwork laid, we are now ready to present the results of our paper.

\section{Same average extractable energy using three kinds of energy extraction processes}
\label{sec3}
In this section, we compare three classes of quantum processes — random unitary operations, random CPTP maps, and general random quantum maps, including NCPTP ones — from the perspective of the average extractable energy from a quantum battery each process yields.

Let $\overline{\mathcal{E}}_{\mathrm{X}}$ denote the average energy extracted from a quantum battery when the extraction is performed using a quantum process $\mathrm{X}$, where $\mathrm{X} = \mathrm{U}$ corresponds to random unitary operations, $\mathrm{X} = \mathrm{CPTP}$ refers to random CPTP maps, and $\mathrm{X} = \mathrm{G}$ means general quantum maps, including physically realizable NCPTP maps, and averaging is performed Haar uniformly over all the maps that belong to the process.
Given a quantum battery characterized by a quantum state $\rho_B$ and Hamiltonian $H_B$, the average extractable energy from the battery using random unitary operations is defined as 
\begin{equation}
 \overline {\mathcal{E}}_{\mathrm{U}} \coloneqq \int dU_B \left [\Tr(\rho_B H_B)-\Tr(U_B \rho_B U^{\dagger}_B H_B)\right],  \label{Eavg1}
\end{equation}
where the unitaries $U_B$ are drawn Haar-randomly from the unitary group $\mathcal{U}(d_B)$, and the integration is performed over the Haar measure (see Appendix~\ref{appA} for details about Haar measure).
To define the average extractable energy from a quantum battery using random CPTP maps or any arbitrary physically realizable quantum maps, we introduce an auxiliary system along with a global unitary that acts jointly on the battery and auxiliary. Let the auxiliary system initially be in the state $\rho_A$ on the Hilbert space $\mathcal{H}_A$, and let the initial joint state of the battery and auxiliary be $\rho_{BA} \in \mathcal{D}(\mathcal{H}_B \otimes \mathcal{H}_A)$, such that $\Tr_A(\rho_{BA}) = \rho_B$.
The average extractable energy using a quantum process $X$ (where $X$ can be random CPTP or general quantum maps, including NCPTP, denoted by $\mathrm{G}$) is defined as
\begin{align}
\label{Eavg2}
&\overline{\mathcal{E}}_{\mathrm{X}} \coloneqq \int \int dU_{BA} d\rho_A \biggl [\Tr\left(\rho_B H_B\right)- \nonumber\\&\Tr\left [\Tr_A\left (U_{BA} \rho_{BA} U^{\dagger}_{BA}\right) H_B\right]\biggr] , 
\end{align}
where the unitaries $U_{BA}$ and the auxiliary state $\rho_A$ are chosen Haar-uniformly from $\mathcal{U}(d_B d_A)$ and $\mathcal{D}(\mathcal{H}^{d_A})$ respectively.
In order to define $\overline{\mathcal{E}}_{\mathrm{CPTP}}$, the initial joint battery-auxiliary state $\rho_{BA}$ must be of the product form, i.e., $\rho_{BA} = \rho_B \otimes \rho_A$. In contrast, to define $\overline{\mathcal{E}}_{\mathrm{G}}$, the joint battery-auxiliary state $\rho_{BA}$ can be any valid quantum state, including entangled ones. 
\begin{theorem} 
\label{th1}
The average extractable energies from a quantum battery
using all three kinds of energy extraction processes – random unitary, random CPTP, and general maps, including physically realizable NCPTP quantum maps – are the same, that is,
\begin{equation}
    \overline{\mathcal{E}}_{\mathrm{U}} = \overline{\mathcal{E}}_{\mathrm{CPTP}} = \overline{\mathcal{E}}_{\mathrm{G}}. \label{Thm1}
\end{equation}
Where $\overline{\mathcal{E}}_{\mathrm{U}}$, $\overline{\mathcal{E}}_{\mathrm{CPTP}}$, and $\overline{\mathcal{E}}_{\mathrm{G}}$ denote the average extractable energy using random unitary, random CPTP, and all physically realizable quantum maps, respectively.
\end{theorem}
The detailed proof of Theorem.~\ref{th1} is given in Appendix~\ref{B1}. We have already mentioned CPTP and NCPTP maps are realized through an auxiliary system and global unitary. In Appendix~\ref{B1} we have proven that the average extractable energy for all three processes is the same, considering all fixed-dimensional auxiliary systems. This basically implies that the average extractable energy for all three processes is also the same if we average over all dimensional auxiliary systems.

\section{Fluctuation of work extraction differ in open and closed dynamics}
\label{sec4}
A machine is deemed efficient when it delivers high performance with minimal fluctuation. Likewise, assessing the performance of a quantum battery requires analyzing not only the average extractable energy but also the consistency of energy extraction. The latter can be effectively measured by fluctuation in extractable energy around the average, for a given energy extraction process. Hence, in this section, we analyze the fluctuations in extractable energy corresponding to the three types of energy extraction processes discussed earlier.

The fluctuation in extractable energy from a quantum battery with an initial state $\rho_B$ and Hamiltonian $H_B$ under random unitary operations is defined as
\begin{align}
(\Delta \mathcal{E}_\mathrm{U})^2 \coloneqq \overline{\mathcal{E}^2}_\mathrm{U} - (\overline{\mathcal{E}}_\mathrm{U})^2, \label{Efluc1}
\end{align}
where the first term is given by
\begin{equation*}
\overline{\mathcal{E}^2}_\mathrm{U} \coloneqq \int dU_B \left[\Tr(\rho_B H_B) - \Tr(U_B \rho_B U_B^\dagger H_B)\right]^2,
\end{equation*}
and $\overline{\mathcal{E}}_\mathrm{U}$ is defined in Eq.~\eqref{Eavg1}.

Similarly, the fluctuation in extractable energy from the battery with an initial state $\rho_B$ and Hamiltonian $H_B$ for energy extraction via random CPTP maps or general quantum maps (including NCPTP), is defined as
\begin{align}
(\Delta \mathcal{E}_\mathrm{X})^2 \coloneqq \overline{\mathcal{E}^2}_\mathrm{X} - (\overline{\mathcal{E}}_\mathrm{X})^2, \label{Efluc2}
\end{align}
where $\mathrm{X}$ corresponds to ${\mathrm{CPTP}, \mathrm{G}}$, and
\begin{align*}
     \overline{\mathcal{E}^2}_\mathrm{X} \coloneqq& \int \int dU_{BA} d\rho_A \biggl [\Tr\left(\rho_B H_B\right)\\&\nonumber- \Tr\left [\Tr_A\left (U_{BA} \rho_{BA} U^{\dagger}_{BA}\right) H_B\right]\biggr]^2.
\end{align*}
The unitaries $U_{BA}$ and auxiliary states $\rho_A$ are sampled in the same way as discussed earlier in Eq.~\eqref{Eavg2}. $\rho_{BA}$ is chosen to be a product state for random CPTP maps, while $\rho_{BA}$ can be any quantum state for arbitrary quantum maps.
\vspace{-4mm}

\subsection{Fluctuation in extractable energy of closed-system batteries}
\label{VA}
We now derive the explicit expression for the fluctuation in extractable energy from a quantum battery under random unitary operations. To proceed, let us consider a $\mathbb{SU}(d_B)$ basis of the Hilbert space $\mathcal{H}_B$ consisting of the matrices $\{\mathbb{I}_{d_B}, \lambda_1, \lambda_2, \ldots, \lambda_{d_B^2 - 1}\}$, where $\{\lambda_i\}$, with $i \in \{1,2, \ldots, (d_B^2 - 1)\}$, are Hermitian, traceless and mutually orthogonal matrices satisfying $\Tr(\lambda_i \lambda_j) = d_B \delta_{ij}$. They are commonly referred to as Gell-Mann matrices. 

\begin{Observation} 
\label{obs1}
Let a quantum battery be described by an initial state $\rho_B$ on a Hilbert space of dimension $d_B$, and a Hamiltonian $H_B$, decomposed as
$H_B = a_0 \mathbb{I}_{d_B} + \sum_{i=1}^{d_B^2 - 1} a_i \lambda_i$,
where $\{\lambda_i\}$ forms an orthogonal $\mathbb{SU}(d_B)$ basis with the identity operator $\mathbb{I}_{d_B}$, and $\vec{a} \in \mathbb{R}^{d_B^2}$ with $\vec{a} \coloneqq \{a_0,a_1, \ldots, a_{d_B^2-1}\}$ are the expansion coefficients. The fluctuation in extractable energy from the battery under random unitary operations is given by
\begin{equation*}
    (\Delta \mathcal{E}_\mathrm{U})^2 = \frac{d_B \alpha_1 - 1}{d_B^2 - 1} \sum_{i=1}^{d_B^2 - 1} a_i^2,
\end{equation*}
where $\alpha_1 (\coloneqq \Tr(\rho_B^2))$ denotes the purity of the state $\rho_B$.
\end{Observation}
The proof of Observation~\ref{obs1} is provided in Appendix~\ref{appObs1}. 

The fluctuation in extractable energy arising from unitary maps is non-zero in general, except when the initial state of the battery is the maximally mixed state $\rho_B=\frac{\mathbb{I}_{d_B}}{d_B}$, for which $\alpha_1=1/d_B$. 

\subsection{Fluctuations in extractable energy of open-system batteries}
\label{VB}
In this section, our goal is to determine the fluctuations in extractable energy from a quantum battery under the action of random CPTP maps and random general quantum maps. Before addressing the fully general case, let us first analyze a restricted scenario: the fluctuations in extractable energy when these maps are realized using an auxiliary system of fixed dimension, along with global unitaries acting on the fixed-dimensional joint battery–auxiliary system. Let us denote the fluctuations in extractable energy from an open-system quantum battery under random CPTP and general quantum maps, implemented using a fixed-dimensional auxiliary system and global unitaries acting on the joint fixed-dimensional battery-auxiliary system,
by $(\Delta\mathcal{E}_{\mathrm{CPTP}})^2_{d_A}$ and $(\Delta\mathcal{E}_{\mathrm{G}})^2_{d_A}$, respectively.

\begin{Observation} 
\label{obs2}
The fluctuations in extractable energy from a quantum battery with an initial state $\rho_B \in \mathcal{D}(\mathcal{H}^{d_B})$ and a Hamiltonian $H_B$ are given by
\begin{align*}
&(\Delta \mathcal{E}_\mathrm{CPTP})^2_{d_A} = \frac{d_B d_A \alpha - 1}{d_B^2 d_A^2 - 1} \sum_{i=1}^{d_B^2 - 1} a_i^2,\\
&(\Delta \mathcal{E}_\mathrm{G})^2_{d_A} = \frac{1}{d_B d_A + 1} \sum_{i=1}^{d_B^2 - 1} a_i^2,
\end{align*}
respectively, when random CPTP maps and random general quantum maps (including physically realizable NCPTP maps), both implemented only using an auxiliary system with a state $\rho_A$ of fixed dimension $d_A$, and global unitaries acting on the Hilbert space $\mathcal{H}^{d_B d_A}$, are used as energy extraction processes. Here, $\{a_i\}$ are the expansion coefficients of the Hamiltonian $H_B$ written in the $\mathbb{SU}(d_B)$ basis, and $\alpha \coloneqq \Tr(\rho_B^2) \int d\rho_A \Tr(\rho_A^2) $.
\end{Observation}
A detailed proof of Observation~\ref{obs2} can be found in Appendix~\ref{appobs2}.

\begin{figure}
\includegraphics[scale=0.23]{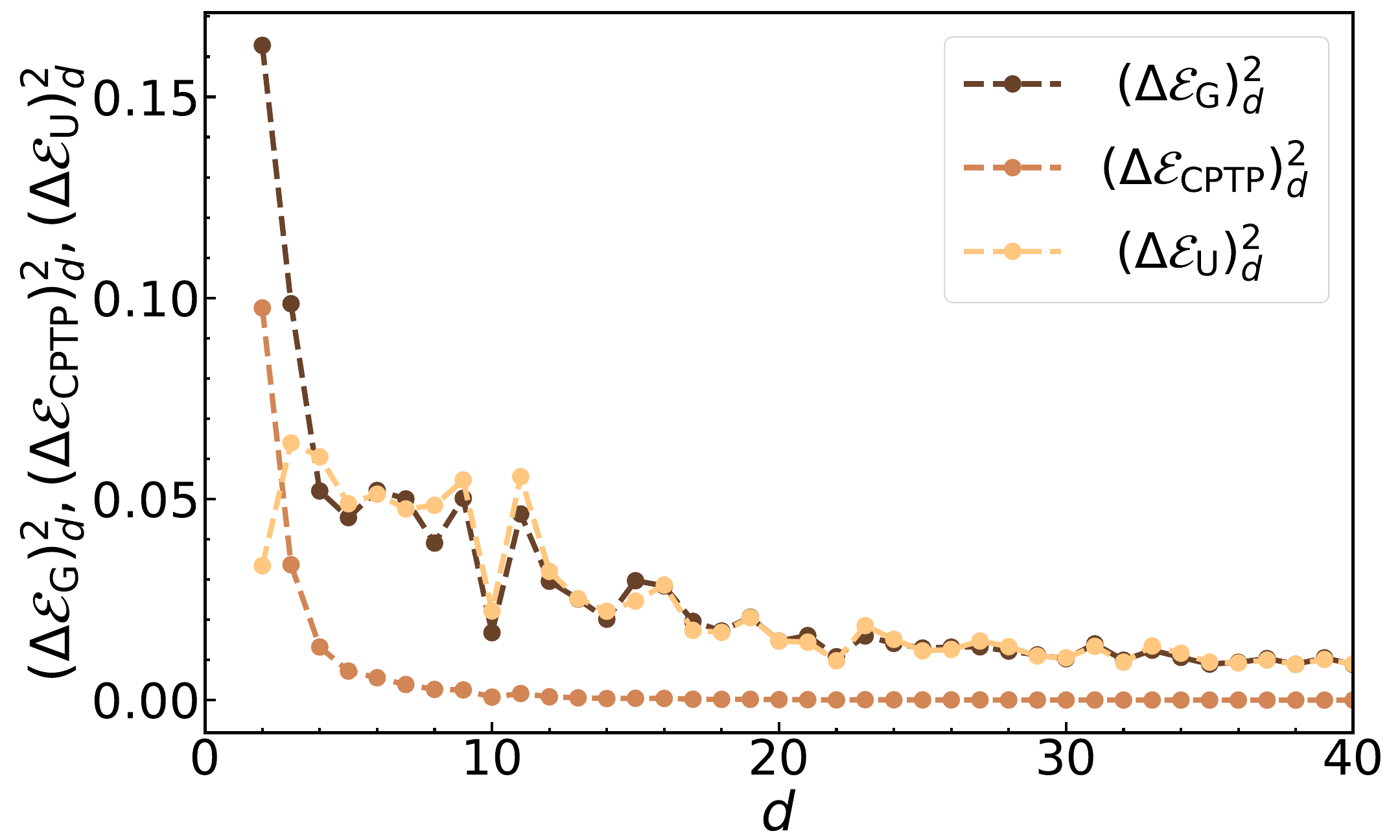}
\caption{\textbf{The behavior of fluctuations in extractable energy produced by unitary maps, CPTP maps, and general maps (including physically realizable NCPTP maps) with respect to the dimension of the battery}. The battery dimension is plotted on the horizontal axis and the fluctuations in extractable energy are plotted on the vertical axis. Here we assume the dimensions of the battery and the auxiliary are the same, i.e., $d_B=d_A=d$. All three kinds of energy extraction processes produce a notable amount of fluctuations in extractable energy in the small limit of $d$. For very small values of $d$, instances are seen where the fluctuation in extractable energy arising from the use of CPTP maps surpasses that from unitary maps, whereas the fluctuation in extractable energy resulting from general quantum maps exceeds its corresponding CPTP counterpart for all values of $d$. As $d$ increases, we can see that the fluctuation in extractable energy produced by CPTP maps becomes smaller compared to the other types of maps. For large values of $d$, fluctuations in extractable energy for all three types of extraction processes tend to zero. The parameters $a_1,a_2,...,a_d $ of the battery Hamiltonian $H_B$ are sampled randomly from uniformly distributed values in the interval $[0,1]$, while $a_{d+1}=a_{d+2}=...=a_{d^2-1}=0$. The vertical axis has the unit of squared energy, while the horizontal axis is dimensionless. }
\label{fig1}
\end{figure}

We numerically compute the fluctuations in extractable energy arising from all the three processes and plot the same with respect to the battery dimension $d_B=d$ in Fig.~\ref{fig1}. We further assume that the dimensions of the battery and the auxiliary are equal, i.e., $d_A=d_B=d$. The parameters $\{a_i\}_{i=1}^{d^2-1}$ of the battery Hamiltonian are taken in the following way: first $d$ coefficients are randomly chosen from the values uniformly distributed in the interval $[0,1]$, while the remaining $d^2-d-1$ coefficients are set to zero. The values of $d$ are taken from $d=2$ to $d=101$. All three types of energy extraction processes exhibit significant fluctuations in extractable energy when $d$ is small. At very low values of $d$, it is observed that the fluctuation due to CPTP maps exceeds that of unitary maps, while the fluctuation from general quantum maps consistently remains higher than that of CPTP maps across all values of $d$. As $d$ grows, the fluctuations from CPTP maps become smaller than those from the other two types of processes. In the limit of large $d$, the extractable fluctuations in extractable energy for all three types of processes diminish to zero.

We now relate the fluctuations in extractable energy from a quantum battery under random CPTP maps with those arising from random general quantum maps (including physically realizable NCPTP maps), where both types of maps are implemented only using an auxiliary system of fixed dimension and global unitaries acting on a fixed-dimensional joint battery-auxiliary system.

\begin{Observation}
\label{obs3}
The fluctuation in extractable energy from an open-system quantum battery, where the extraction process is a random CPTP map implemented using a fixed-dimensional auxiliary system and global unitaries acting on the joint fixed-dimensional battery-auxiliary system, is always less than or equal to the fluctuation in extractable energy when the extraction process is a random general quantum map (including physically realizable NCPTP maps), realized using the same dimensional auxiliary system and global unitaries acting on the same dimensional joint battery-auxiliary system.
\end{Observation}

\begin{proof}
$\alpha \coloneqq \Tr(\rho^2_B)\int \Tr(\rho^2_A) d\rho_A$ (defined in Observation~\ref{obs2}) represents the average of the purities of all possible states of auxiliaries corresponding to a fixed dimension, multiplied with the purity of the battery's initial state. Since the purity of any quantum state is always upper bounded by unity and the average of any physical quantity is always less than or equal to the maximum value that the quantity can have, $\alpha$ is also upper bounded by unity. As $\alpha \leq 1$, we have
\begin{align*}
(\Delta\mathcal{E}_{\mathrm{CPTP}})^2_{d_A}&=\frac{d_Bd_A\alpha-1}{d_B^2d_A^2-1} \sum_{i=1}^{d_B^2-1} a_i^2\leq\frac{d_Bd_A-1}{d_B^2d_A^2-1} \sum_{i=1}^{d_B^2-1} a_i^2,\\
&=\frac{1}{d_Bd_A+1} \sum_{i=1}^{d_B^2-1} a_i^2=(\Delta \mathcal{E}_{\mathrm{G}})^2_{d_A}. 
\end{align*}
This concludes the proof.
\end{proof}

\textit{Remark:}
Note that the fluctuations in extractable energy, $(\Delta\mathcal{E}_{\mathrm{CPTP}})^2_{d_A}$ and $(\Delta\mathcal{E}_{\mathrm{G}})^2_{d_A}$ are computed under random CPTP maps and general quantum maps realized using auxiliary systems of fixed dimension $d_A$.

We will determine the fluctuations in extractable energy from a quantum battery under random CPTP maps and general quantum maps, $(\Delta\mathcal{E}_{\mathrm{CPTP}})^2$ and $(\Delta\mathcal{E}_{\mathrm{G}})^2$ respectively, by averaging $(\Delta\mathcal{E}_{\mathrm{CPTP}})^2_{d_A}$ and $(\Delta\mathcal{E}_{\mathrm{G}})^2_{d_A}$ over all finite dimensional auxiliary systems and the corresponding global unitaries acting on the joint battery and auxiliary system.

\begin{theorem} 
\label{th2}
The fluctuations in extractable energy from a quantum battery vanish when the energy extraction process is implemented by using either random CPTP quantum maps or random general quantum maps (including all physically realizable NCPTP maps), i.e., $ (\Delta \mathcal{E}_{\mathrm{CPTP}})^2 = 0 \quad \mathrm{and} \quad (\Delta \mathcal{E}_{\mathrm{G}})^2 = 0$. 
\end{theorem}

\begin{proof}
Since evaluating the fluctuation in extractable energy from a quantum battery under random general quantum maps requires averaging over all possible dimensions of auxiliary systems and the corresponding global unitaries acting on all finite dimensional joint battery and auxiliary system, the fluctuation in extractable energy under random general quantum maps can be written as
\begin{align*}
    (\Delta \mathcal{E}_{\mathrm{G}})^2=\lim_{n \rightarrow \infty} (\Delta\mathcal{E}_{\mathrm{G}})^2_n,
\end{align*}
where $(\Delta\mathcal{E}_{\mathrm{G}})^2_n \coloneqq \frac{1}{n} \sum_{d_A=2}^{n+1} (\Delta\mathcal{E}_{\mathrm{G}})^2_{d_A}$. From the expression of  $(\Delta\mathcal{E}_{\mathrm{G}})^2_{d_A}$ given in Observation~\ref{obs2}, it can be seen that $\sum_{i=1}^{d_B^2-1} a_i^2$ is independent of $d_A$. Therefore we can write $(\Delta \mathcal{E}_{\mathrm{G}})^2$ as follows:
\begin{eqnarray*}
\label{all_N}
    (\Delta \mathcal{E}_{\mathrm{G}})^2=\left[\lim_{n \rightarrow \infty} \frac{1}{n}\left(\sum_{d_A=2}^{n+1}\frac{1}{d_Bd_A+1}
\right)\right]\sum_{i=1}^{d_B^2-1} a_i^2.
\end{eqnarray*}

Now, note that the sequence $\left\{\frac{1}{d_Bd_A+1}\right\}$ forms a subsequence of the harmonic sequence $\left\{\frac{1}{m}\right\}$, where $m$ denotes any positive integer. As the latter converges to zero, the former too
converges to zero. This is due to the fact that all the subsequences of a convergent sequence with a limit $l$ converge to the same limit $l$~\cite{knopp}. 
Cauchy's limit Theorem~\cite{knopp}, states that if a sequence $\left\{c_l\right\}$ converges to a limit $l$, the corresponding sequence of the arithmetic mean of the first $n$ members of the sequence also converges to the same limit $l$. 
Hence, the sequence $\left\{\frac{1}{n}\left(\sum_{d_A=2}^{n+1}\frac{1}{d_Bd_A+1} \right)\right\}$, which is the sequence of the arithmetic mean of the first $n$ members of the sequence $\left\{\frac{1}{d_Bd_A+1}\right\}$, also converges to zero.
Therefore, $ (\Delta \mathcal{E}_{\mathrm{G}})^2=0$.

Since Observation~\ref{obs3} states that $(\Delta \mathcal{E}_{\mathrm{CPTP}})^2_{d_A} \leq (\Delta \mathcal{E}_{\mathrm{G}})^2_{d_A}$ for \emph{all} $d_A$, it follows that $(\Delta \mathcal{E}_{\mathrm{CPTP}})^2 \leq (\Delta \mathcal{E}_{\mathrm{G}})^2$.  
But $ (\Delta \mathcal{E}_{\mathrm{G}})^2=0$, from which it follows that  $(\Delta \mathcal{E}_{\mathrm{CPTP}})^2\leq0$.
On the other hand, $ (\Delta \mathcal{E}_{\mathrm{CPTP}})^2 \geq 0 $, as it represents the variance of the extractable energy, and the variance of any physical quantity is always non-negative. 
Together, these inequalities imply that $(\Delta \mathcal{E}_{\mathrm{CPTP}})^2=0$. 
This completes the proof.
\end{proof}

Therefore, open-system batteries offer greater advantages over closed ones when it comes to minimizing fluctuations in energy extraction.

\subsection{Trade-off in between resource-cost and the performance of open-system batteries}
Both CPTP maps and arbitrary quantum maps processes exhibit zero fluctuation in extractable energy. 
To make a comparison of the robustness of these two processes, we consider the maps realized by auxiliary systems up to a fixed highest dimension $n$ instead of considering all possible dimensional auxiliary systems, and see how the fluctuation in extractable energy decays with increasing $n$ for CPTP and any general maps.

\begin{theorem} \label{th3}
The fluctuation of extractable energy arising from CPTP maps falls off faster with $n$, the highest dimension of the auxiliary among all the considered auxiliaries, than that resulting from arbitrary quantum maps, with the scaling of the former being $1/n$ while that of the latter is $\ln n/n$.
\end{theorem}

The proof is given in Appendix.~\ref{appth3}.

The implication of Theorem~\ref{th3} is that in the large $n$ regime, the fluctuation of extractable energy from quantum battery via CPTP maps decreases more rapidly with the highest dimension of the auxiliary among all the considered auxiliaries, i.e $n$, compared to the fluctuation in extractable energy arising from general quantum maps. Therefore, in the large $n$ regime, it is preferable to use CPTP maps over the general quantum maps, as the fluctuation caused by the CPTP map decreases more sharply with $n$.

\textit{Remark:} Note that we analyze the scaling behavior of fluctuations in extractable energy with respect to $n$, where $n$ denotes the maximum dimension up to which auxiliary systems are considered for realizing CPTP and arbitrary quantum maps.

\begin{figure*}
\includegraphics[scale=0.22]{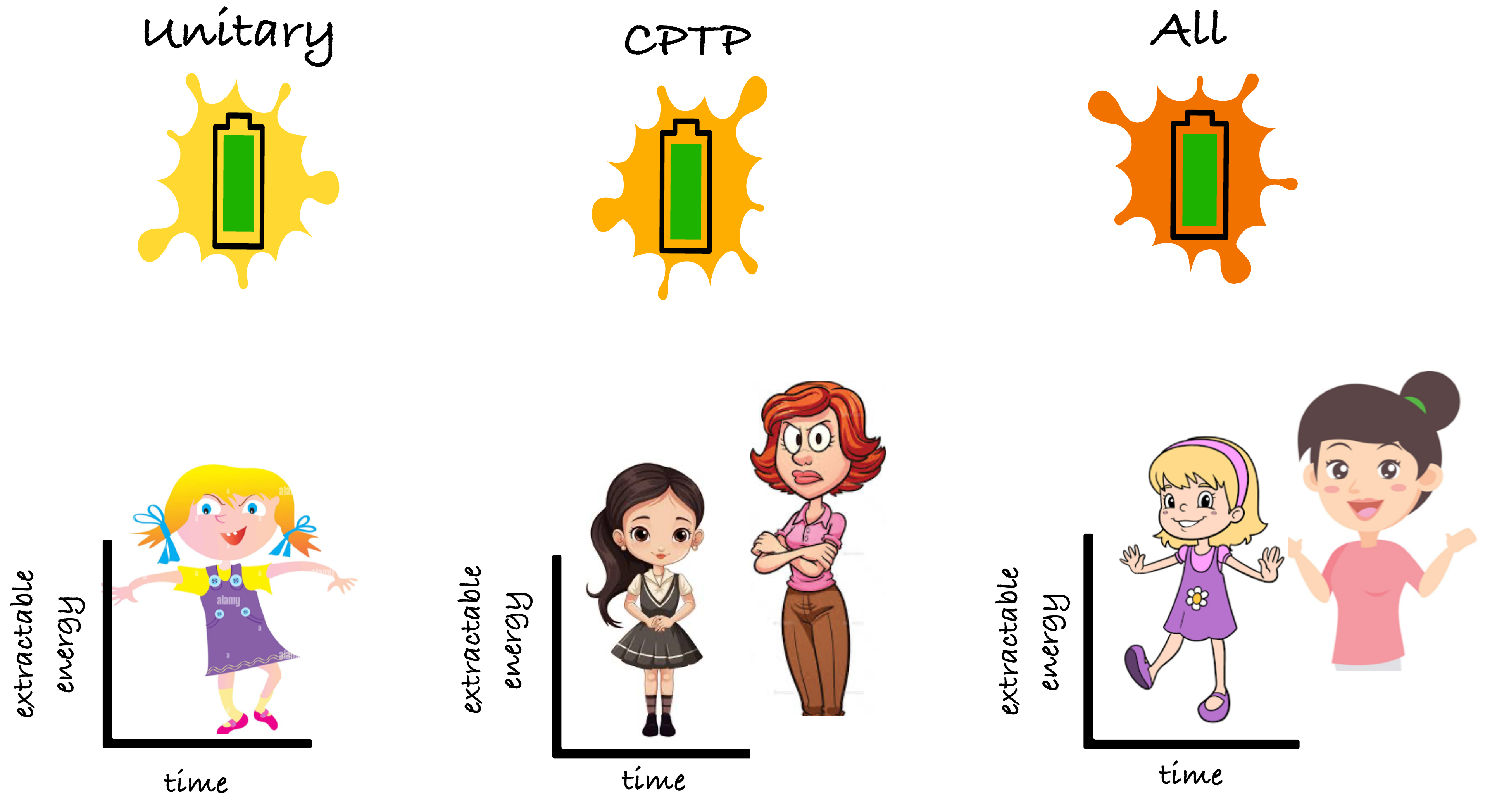}
\caption{\textbf{Schematic representation of the fluctuations in extractable energy from a quantum battery under random unitary, random CPTP, and random general quantum maps.} 
We have shown that the average extractable energy from the battery (depicted as a green rectangle with black border) for all three kinds of energy extraction processes considered is the same, while the fluctuations in extractable energy 
are nonzero exclusively under random unitary maps (described by yellow paint splatter). In contrast, when energy is extracted from the battery using random CPTP maps (golden yellow splatter) or general quantum maps, including physically realizable non-CPTP maps (chrome yellow splatter), the fluctuations in extractable energy vanish. Moreover, the fluctuations in extractable energy decrease more rapidly with the auxiliary system’s maximum dimension under random CPTP maps compared to under general quantum maps.
Our result can be visualized through a real-life scenario of three different days of a little girl: 
on first day, a little girl is home alone without her parents, free to do whatever she wants, represented by the leftmost image of a single girl.  On second day, the girl is under the strict supervision of her mother, depicted in the middle, showing the girl with her mother, while on third day, the girl is with the loving care of her mother, depicted in the lower right corner, where she is neither completely restricted, as in the situation depicted in the middle figure, nor entirely free, as in the leftmost case.
Note that, girl in home alone without her parents, girl under the strict supervision of her mother, girl with the loving care of her mother, represents the fluctuations in extractable energy from a quantum battery under random unitary, random CPTP, and random general quantum maps respectively. 
}
\label{fig2}
\end{figure*}

In the next section, we compare the fluctuations in extractable energy corresponding to the considered extraction processes realized with a fixed auxiliary dimension.

\section{The case of fixed auxiliary dimension}
\label{sec5}
In this section, we figure out the possible situations depending on the hierarchy of the fluctuations in extractable energy from the quantum battery due to the three energy extraction processes: random unitary, random CPTP, and any general maps, including physically realizable NCPTP maps.

We first identify the condition under which the fluctuation in extractable energy associated with arbitrary quantum maps surpasses that of unitary maps, i.e.,
$(\Delta\mathcal{E}_{\mathrm{G}})_{d_A}^2>(\Delta\mathcal{E}_{\mathrm{U}})_{d_A}^2$. This implies that
$\frac{(\Delta\mathcal{E}_{\mathrm{G}})_{d_A}^2}{(\Delta\mathcal{E}_{\mathrm{U}})_{d_A}^2}=\frac{d^2_B-1}{\left(d_B\alpha_1-1\right)\left(d_Ad_B+1\right)}>1$. 
A straightforward analysis reveals that the following condition must be met for the inequality to be valid, given by 
\begin{equation}
    \alpha_1<\frac{d_A+d_B}{1+d_Ad_B}. \label{cond1}
\end{equation}
Thus, for a given auxiliary dimension $d_A$, arbitrary quantum maps give rise to higher amount of fluctuation in extractable energy as compared to unitary maps whenever the purity of the initial state of battery satisfies the condition~\eqref{cond1}.

Let us now examine the condition under which the fluctuation in extractable energy arising from unitary maps exceeds that resulting from CPTP maps, i.e.,
 $(\Delta\mathcal{E}_{\mathrm{U}})_{d_A}^2>(\Delta\mathcal{E}_{\mathrm{CPTP}})_{d_A}^2$. This implies that $\frac{(\Delta\mathcal{E}_{\mathrm{U}})_{d_A}^2}{(\Delta\mathcal{E}_{\mathrm{CPTP}})_{d_A}^2}=\frac{d_B\alpha_1-1}{d_B^2-1}\frac{d_A^2d_B^2-1}{d_Ad_B\alpha-1}>1$. Denoting $\beta_1=\int \Tr\left(\rho_A^2\right) d\rho_A$ and noting that $\alpha_1=\Tr(\rho^2_B)$, we can write $\alpha=\Tr(\rho^2_B)\int \Tr\left(\rho_A^2\right) d\rho_A=\alpha_1\beta_1$. Therefore in terms of the variables $\alpha_1$ and $\beta_1$, the inequality $\frac{(\Delta\mathcal{E}_{\mathrm{U}})_{d_A}^2}{(\Delta\mathcal{E}_{\mathrm{CPTP}})_{d_A}^2}>1$ leads to the following condition:
\begin{equation}
    \alpha_1>\frac{d_B\left(d_A^2-1\right)}{d_A^2d_B^2-1-\beta_1 d_A\left(d_B^2-1\right)}.\label{cond2}
\end{equation}
Thus,~\eqref{cond2} gives the condition under which fluctuation in extractable energy induced by unitary maps surpass that resulting from CPTP maps. Hence, by making use of Observation~\ref{obs3} and the conditions~\eqref{cond1} and~\eqref{cond2}, we are led to the following three scenarios:
\begin{itemize}
    \item $(\Delta \mathcal{E}_{\mathrm{G}})_{d_A}^2>(\Delta \mathcal{E}_{\mathrm{U}})_{d_A}^2>(\Delta \mathcal{E}_{\mathrm{CPTP}})_{d_A}^2$,
     \item $(\Delta \mathcal{E}_{\mathrm{G}})_{d_A}^2>(\Delta \mathcal{E}_{\mathrm{CPTP}})_{d_A}^2>(\Delta \mathcal{E}_{\mathrm{U}})_{d_A}^2$,
      \item $(\Delta \mathcal{E}_{\mathrm{U}})_{d_A}^2>(\Delta \mathcal{E}_{\mathrm{G}})_{d_A}^2>(\Delta \mathcal{E}_{\mathrm{CPTP}})_{d_A}^2$.
\end{itemize}

However, in the case of large dimensional battery, the unitary maps always generate a higher fluctuation in the extractable energy compared to the CPTP maps. To show this, let us work out the quantity $ \frac{(\Delta\mathcal{E}_{\mathrm{U}})_{d_A}^2-(\Delta\mathcal{E}_{\mathrm{CPTP}})_{d_A}^2}{\sum_{i=1}^{d^2_B-1}a^2_i}$, where we assume $d_B$ is sufficiently large such that $\frac{1}{d_B^2}$ and higher-order terms can be neglected. No restriction is placed on the value of auxiliary dimension $d_A$, but it is assumed to be fixed. Thus we have 
\begin{align}
    \frac{(\Delta\mathcal{E}_{\mathrm{U}})_{d_A}^2-(\Delta\mathcal{E}_{\mathrm{CPTP}})_{d_A}^2}{\sum_{i=1}^{d^2_B-1}a^2_i}&=\frac{d_B\alpha_1-1}{d^2_B-1}-\frac{d_Ad_B\alpha-1}{d^2_Ad^2_B-1},\nonumber\\
    &\approx \frac{\left(d_A\alpha_1-\alpha\right)}{d_Bd_A}. \label{ucptpmain}
\end{align}
Now, the quantity $\beta_1=\int \Tr\left(\rho_A^2\right) d\rho_A\leq1$ as it is the averaged purity of the states of all auxiliaries with a given dimension, and the purity of any quantum state has its maximum value of unity. Using the relation $\alpha=\alpha_1\beta_1$, it follows that $\alpha\leq\alpha_1$. Thus, we find that the RHS of Eq.~\eqref{ucptpmain} is nonnegative, which proves our claim.

Let us now perform the same analysis for fluctuations in extractable energy resulting from unitary maps $(\Delta\mathcal{E}_{\mathrm{U}})_{d_A}^2$ and arbitrary quantum maps $(\Delta\mathcal{E}_{\mathrm{G}})_{d_A}^2$. Taking the difference, we have
\begin{align}
    \frac{(\Delta\mathcal{E}_{\mathrm{G}})_{d_A}^2-(\Delta\mathcal{E}_{\mathrm{U}})_{d_A}^2}{\sum_{i=1}^{d^2_B-1}a^2_i}&=\frac{1}{d_Bd_A+1}-\frac{d_B\alpha_1-1}{d^2_B-1},\nonumber\\
    &\approx\frac{1}{d_B}\left[\frac{d_A+d_B}{1+d_Ad_B}-\alpha_1\right]. 
    \label{uncptpmain}
\end{align}
From Eq.~\eqref{uncptpmain}, we can immediately see that the RHS is nonnegative for $\alpha_1<\frac{d_A+d_B}{1+d_Ad_B}$, which is simply Eq.~\eqref{cond1}. Hence, we cannot rule out either of the cases $(\Delta \mathcal{E}_{\mathrm{G}})_{d_A}^2>(\Delta \mathcal{E}_{\mathrm{U}})_{d_A}^2$ or $(\Delta \mathcal{E}_{\mathrm{U}})_{d_A}^2>(\Delta \mathcal{E}_{\mathrm{G}})_{d_A}^2$. 
Consequently, for higher-dimensional batteries, i.e., in the regime of high $d_B$, we are left with only two situations, which are given below.
\begin{itemize}
   \item $(\Delta \mathcal{E}_{\mathrm{G}})_{d_A}^2>(\Delta \mathcal{E}_{\mathrm{U}})_{d_A}^2>(\Delta \mathcal{E}_{\mathrm{CPTP}})_{d_A}^2$,
    \item $(\Delta \mathcal{E}_{\mathrm{U}})_{d_A}^2>(\Delta \mathcal{E}_{\mathrm{G}})_{d_A}^2>(\Delta \mathcal{E}_{\mathrm{CPTP}})_{d_A}^2$.
\end{itemize}
See Appendix~\ref{appC} for the detailed steps leading to Eqs.~\eqref{ucptpmain} and~\eqref{uncptpmain}.
An illustration of the various possible scenarios resulting from a comparison based on the magnitude wise ordering of the fluctuations in extractable energy in the moderate and high $d_B$ regions is depicted in Fig.~\ref{fig4}.

Next, we examine the asymptotic scaling of the fluctuations in extractable energy generated by the use of CPTP maps and general quantum maps with the auxiliary dimension $d_A$, for a given battery. From Observation~\ref{obs2}, $(\Delta \mathcal{E}_{\mathrm{G}})^2_{d_A} = \frac{1}{d_B d_A + 1} \sum_{i=1}^{d_B^2 - 1} a_i^2$. It is evident that $(\Delta\mathcal{E}_{\mathrm{G}})_{d_A}^2 \sim \frac{1}{d_A}$ in the asymptotic limit of $d_A$. This follows from the fact that $\sum_{i=1}^{d_B^2-1} a_i^2$ is independent of $d_A$. To determine the same for CPTP maps, we can write the expression for the corresponding fluctuation in extractable energy as $(\Delta\mathcal{E}_{\mathrm{CPTP}})_{d_A}^2=\frac{d_Ad_B\alpha-1}{d^2_Ad^2_B-1}\approx \frac{\alpha}{d_Ad_B}-\frac{1}{d_A^2d_B^2}$, in the limit of large $d_A$. Here we note that $\alpha$ is a function of $d_A$ (Eq.~\eqref{alpha}) and this dependence must be taken into account when determining the scaling. In particular, $\alpha=\frac{2d_A}{d_A^2\left(1+\frac{1}{d_A^2}\right)}\Tr\left(\rho_B^2\right)\approx\frac{2}{d_A}\Tr\left(\rho_B^2\right)$ in the asymptotic regime. Since $\Tr\left(\rho_B^2\right)$ is independent of $d_A$, we can write $(\Delta\mathcal{E}_{\mathrm{CPTP}})_{d_A}^2\approx \frac{2\Tr\left(\rho_B^2\right)}{d^2_Ad_B}-\frac{1}{d_A^2d_B^2}$. The asymptotic scaling of the fluctuation in extractable energy produced by CPTP maps is now easily determined to be $(\Delta\mathcal{E}_{\mathrm{CPTP}})_{d_A}^2 \sim \frac{1}{d^2_A}$. 

Hence, for a given battery, the fluctuations in energy extraction for CPTP maps and arbitrary maps realized through a fixed-dimensional auxiliary $d_A$ go like $1/d_A^2$ and $1/d_A$, respectively.

\section{Scaling with battery dimension}
\label{sec6}
In this section, we inspect the scaling of the fluctuations in extractable energy resulting from all types of maps with the battery dimension.

To determine the scaling of the fluctuations in extractable energy with the dimension of the battery, we assume that $\sum_{i=1}^{d_B^2-1} a_i^2$ is constant, independent of $d_B$. Then, the asymptotic scaling (at higher values of $d_B$) of the fluctuations for three types of processes for a given auxiliary dimension $d_A$ is given by
\begin{equation}
   (\Delta \mathcal{E}_{\mathrm{U}})_{d_A}^2,(\Delta \mathcal{E}_{\mathrm{CPTP}})_{d_A}^2,(\Delta \mathcal{E}_{\mathrm{G}})_{d_A}^2\sim \frac{1}{d_B}.\nonumber
\end{equation}
Since the fluctuations in extractable energy scale as $1/d_B$ for all three process classes when the auxiliary system has a fixed dimension, the fluctuations in extractable energy under the action of CPTP maps and general quantum maps - averaged over all possible auxiliary system dimensions -  will exhibit the same $1/d_B$ scaling in the large $d_B$ regime, i.e.,
\begin{equation*}
   (\Delta \mathcal{E}_{\mathrm{U}})^2,(\Delta \mathcal{E}_{\mathrm{CPTP}})^2,(\Delta \mathcal{E}_{\mathrm{G}})^2\sim \frac{1}{d_B}.
\end{equation*}
Consequently, the fluctuations in extractable energy vanish in the limit $d_B \to \infty$ for all three process classes.

\begin{figure}
\includegraphics[scale=0.165]{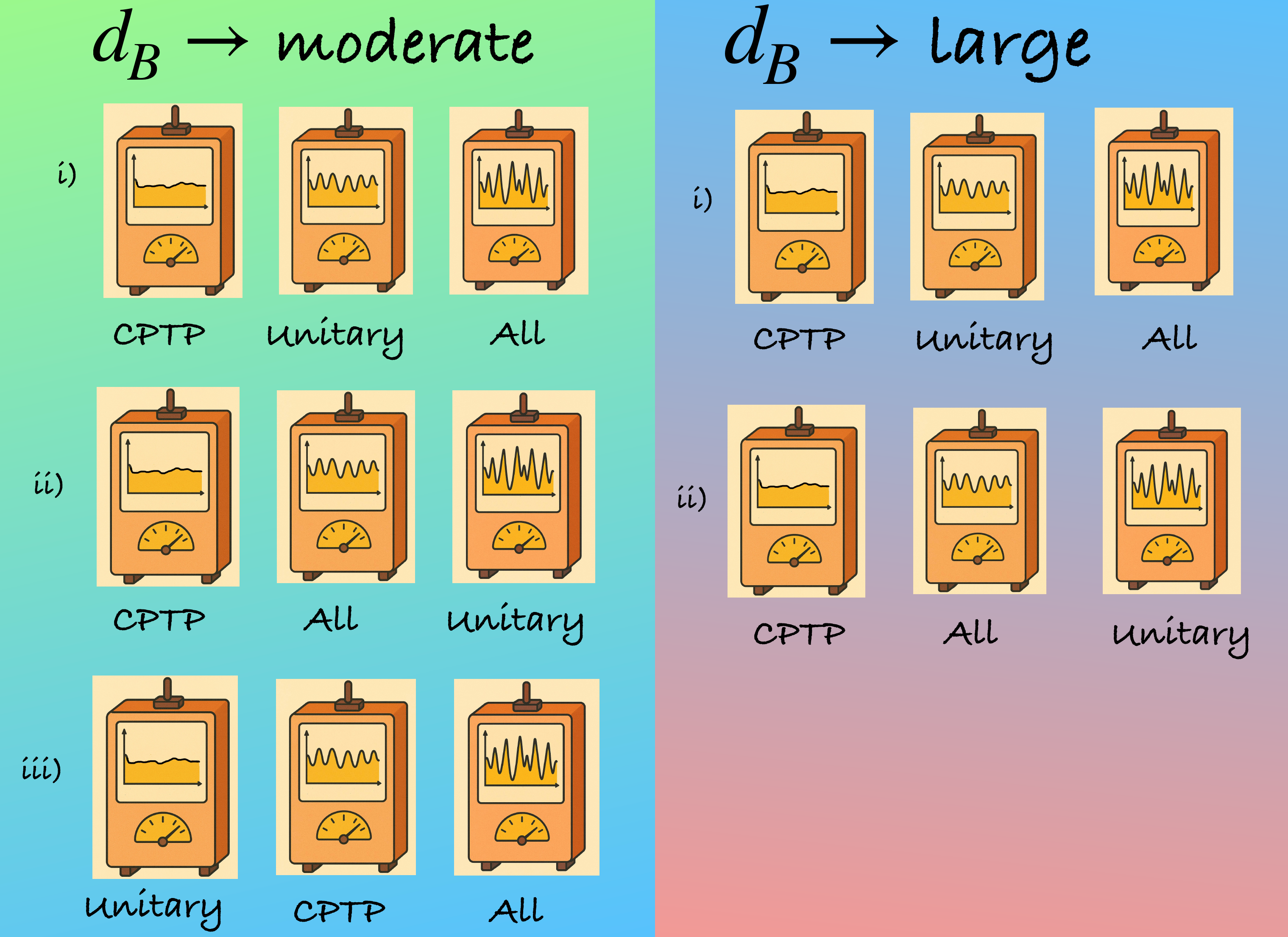}
\caption{\textbf{Schematic representation of the various possible scenarios based on comparing the fluctuations in extractable energy resulting from all three kinds of energy extraction processes.}
The readings of the meter boxes represent the magnitude of the fluctuations in extractable energy, while the labels below each box denote the type of the quantum map used in energy extraction. A higher degree of "rippling" in the readings indicates a greater magnitude of the fluctuation. The left side of the partition depicts the situation when the battery dimension ranges from small or moderate values. In this region, unitary maps can generate higher as well as lower amount of fluctuation in extractable energy in comparison with the rest of the maps. The right side of the partition shows the regime of high battery dimension. Here, unitary maps always produce a larger amount of fluctuation in extractable energy in comparison with CPTP maps. Since CPTP maps are proven to generate a lesser (or equal) amount of fluctuation in extractable energy as compared to general quantum maps for all possible values of the dimension of the battery, this is common to both sides of the partition of the figure.}
\label{fig4}
\end{figure}

From the above discussion, it is clear that for high-dimensional batteries—such as those composed of many-body systems—all three processes exhibit the same
robustness. This is because the scaling of fluctuations in extractable energy with respect to battery dimension is the same across all cases. Moreover, since the average extractable energy remains the same for all three processes, unitary maps are favorable for high-dimensional battery systems, as they require the least resources among the three classes of energy extraction processes.

\vspace{-2mm}
\section{Conclusion}\label{sec7}
Quantum thermal devices have become one of the key areas of research in present-day quantum technologies. Although these devices are expected to outperform their classical counterparts, they face a major challenge due to the presence of unavoidable fluctuations in the associated thermodynamic quantities. Thus, a successful realization of these machines requires finding suitable methods to minimize the fluctuations.

Quantum batteries constitute an important example of quantum thermal devices. The main function of a quantum battery is to store energy, which can be extracted and put to use in various tasks. Different types of quantum maps, viz. unitary maps, completely positive trace-preserving (CPTP) maps, and noncompletely positive trace-preserving (NCPTP) maps, can be used for the purpose of energy extraction.

The main aim of this work has been to determine, both qualitatively and quantitatively, which of these processes should be employed to achieve the best possible performance from a quantum battery with the least fluctuations. As a measure of efficiency, we considered the average extractable energy and the corresponding fluctuations in extractable energy.

We used randomly chosen unitary maps, CPTP maps, and arbitrary quantum maps (including NCPTP maps). We found that an equal amount of energy can be extracted on average, using all three types of processes. Henceforth, we calculated the fluctuations in extractable energy for all processes. It was seen that a nonzero amount of fluctuation persists in the case of unitary maps (except when the battery begins in a maximally mixed state); while employing CPTP maps and general quantum maps, the fluctuations in extractable energy become zero.

Note that CPTP maps and arbitrary quantum maps are implemented using an auxiliary system when necessary. 
Restricting the dimension of the auxiliary up to a finite value $n$, we estimated the scaling of the finite-term averages of the fluctuations in extractable energy for CPTP maps and arbitrary quantum maps over the considered auxiliary dimensions with respect to $n$. The averaged fluctuation in extractable energy generated by CPTP maps is found to scale as $1/n$, which falls off to zero at a faster rate compared to the scaling of $\ln n/n$ of the averaged fluctuation in extractable energy produced by arbitrary quantum maps.

Moreover, we investigated the characteristics of the fluctuations in extractable energy for the three processes while the battery and the auxiliary dimension are kept fixed. It was seen that the fluctuation in extractable energy arising from the use of CPTP maps as the energy extraction process, implemented with a fixed dimensional auxiliary, is upper bounded by the corresponding fluctuation in extractable energy when arbitrary quantum maps, having the same dimension of the auxiliary system as in the case of CPTP maps, are employed in the energy extraction process. Fluctuation in extractable energy produced by unitary maps, however, can be higher as well as lower, as compared to the fluctuations in extractable energy arising from the use of CPTP maps and arbitrary quantum maps. We specifically derived the conditions under which the fluctuation in extractable energy generated by unitary maps exceeds the corresponding fluctuations in extractable energy generated by CPTP and arbitrary quantum maps. Next, we found the scaling of fluctuations in extractable energy with the auxiliary dimension $d_A$ for CPTP and arbitrary quantum maps realized by auxiliaries with a fixed dimension $d_A$. Fluctuations in extractable energy for CPTP maps and arbitrary quantum maps scale as $1/d_A^2$ and $1/d_A$, respectively, for a given battery.

Furthermore, we inspected the large-scale behavior of the fluctuations in extractable energy arising from all three classes of processes with the battery dimension $d_B$. The fluctuations for all three types of processes vary as $1/d_B$. Hence, for quantum batteries with large dimension, unitary maps are just as effective as more resource-intensive open maps in minimizing fluctuation, making unitary maps more preferable choice.

Our findings highlight several key takeaways. Open-system evolution can significantly enhance the performance of a fixed quantum battery compared to closed (unitary) dynamics. Moreover, when open evolution is approached either via CPTP maps or more general quantum maps—including physically realizable non-CPTP maps—CPTP maps are found to outperform the broader class, even though non-CPTP operations typically require greater resources for their implementation. This counter-intuitive result suggests a fundamental trade-off between the performance of a quantum battery and the resource cost associated with implementing extraction protocols.
On the other hand, unitary dynamics can match the low fluctuations achieved by more resource-intensive open processes (CPTP and general quantum maps) for sufficiently large batteries. Unitary maps become more favorable, offering comparable performance with reduced implementation cost, for large quantum batteries. These insights deepen our understanding of the performance of a quantum battery from a statistical perspective under different types of extraction processes.

\section*{Acknowledgment}
We acknowledge the use of Armadillo and QIClib – a modern C++ library for quantum information and computation (\url{https://titaschanda.github.io/QIClib}). 
PG acknowledges support from the ‘INFOSYS scholarship for senior students’ at Harish-Chandra Research Institute, India.

\bibliography{ref}
\onecolumngrid

\section*{Appendix}
\appendix
\setcounter{figure}{0}

\section{Important concepts and formulae} \label{appA}
The Haar measure~\cite{Mele2024introductiontohaar} provides the unique probability measure on the unitary group $\mathcal{U}(d)$.  It is invariant under both left and right shift of elements of the unitary group. That is, for $U,V\in \mathcal{U}(d)$ and any integrable function $f$, $\int dU f(U)= \int f(UV) dU=\int f(VU) dU$. Being a probability measure, it is normalized as $\int dU=1$. 
For any operator $O \in \mathcal{L}(\mathbb{C}^d)^{\otimes k}$, the unitary integral $ I^k(O)=\int dU U^{\otimes k} O(U^{\dagger})^{\otimes k}$
taken over the Haar measure can be explicitly computed~\cite{Rivas_2012,matrix-integral,Roberts2017,PRXQuantum.2.010201}. In particular, for $k=1,2$,

\begin{align}
\label{A1}
    &\int_{\mathcal{U}(d)} dU\, U O U^{\dagger} = \Tr(O)\, \frac{\mathbb{I}_d}{d}, \quad \int_{\mathcal{U}(d)} dU\, U^{\otimes 2} O (U^{\dagger})^{\otimes 2}
    = \frac{1}{d^2 - 1} \left[ \left( \Tr(O) - \frac{\Tr(\mathbb{S} O)}{d} \right) \mathbb{I}_d^{\otimes 2} + \left( \Tr(\mathbb{S} O) - \frac{\Tr(O)}{d} \right) \mathbb{S} \right].
\end{align}

The operator $\mathbb{S}$ appearing in Eq.~\eqref{A1} is the SWAP operator, defined on the operator space $\mathcal{L}\left(\mathcal{H}^d \otimes \mathcal{H}^d\right)$. It can be written as $\mathbb{S}=\sum_{i,j=0}^{d-1} \ket{ij}\bra{ji}$, 
where the set $\{\ket{i}\}$ forms an orthonormal basis of $\mathcal{H}^d$.
For two operators $P, Q \in \mathcal{L}(\mathcal{H}^d)$, 
\begin{equation}
    \Tr[(P \otimes Q)\mathbb{S}]=\Tr[PQ]. \label{A3}
\end{equation} 

Let $S\in \mathcal{L}(\mathcal{H}^{d_1}_A\otimes\mathcal{H}^{d_1}_B)$ and $S_1\in \mathcal{L}\left(\left(\mathcal{H}^{d_1}_A\otimes\mathcal{H}^{d_2}_C\right)\otimes \left(\mathcal{H}^{d_1}_B\otimes\mathcal{H}^{d_2}_D\right)\right)$ be two SWAP operators. The following identity relates $S_1$ and $S$,
\begin{equation}
    \Tr_{C,D}S_1=d_2S. \label{swap1}
\end{equation}

\section{Proofs of Observations~\ref{obs1},~\ref{obs2} and Theorems~\ref{th1} and~\ref{th3}}\label{appB}
Here we give the explicit calculations for the Observations~\ref{obs1},~\ref{obs2} and Theorems~\ref{th1},~\ref{th3} in chronological order. 
\subsection{Calculations of Theorem~\ref{th1}} \label{B1}
The extractable energy from a quantum battery with a state $\rho_B$ and a Hamiltonian $H_B$ can be written as $\mathcal{E}=E-E'$, where $E=\Tr\left(\rho_BH_B\right)$ and $E'=\Tr\left(\rho'_BH_B\right)$ describe the energies of the battery before and after the energy extraction process. In particular, for energy extraction involving unitary maps, $E'_\mathrm{U}\coloneqq \Tr\left(U_B\rho_BU^{\dagger}_BH_B\right)$. Therefore, the average unitarily extractable energy obtained from the battery is given as $ \overline{\mathcal{E}}_{\mathrm{U}}=\overline{E}_{\mathrm{U}}-\overline{E'}_{\mathrm{U}}
    =\Tr\left(\rho_BH_B\right)-\int dU_B\Tr\left(U_B\rho_BU^{\dagger}_B\right)$. Here $\overline{E}_{\mathrm{U}}=E=\Tr\left(\rho_BH_B\right)$
since the initial energy of the battery is independent of the particular process used for energy extraction and therefore a constant. We evaluate the second term using Eq.~\eqref{A1} as,
\begin{align}
   \overline{E'}_{\mathrm{U}}&=\int  dU_B  \Tr\left(\rho'_BH_B\right)=  \Tr\left[\left(\int  dU_BU_B \rho_B U^{\dagger}_B \right)H_B\right] 
   = \Tr\left[\left(\Tr(\rho_B) \frac{\mathbb{I}_{d_B}}{d_B}\right)H_B\right]=\frac{\Tr(H_B)}{d_B}. \label{u4}
\end{align}
 Hence, the average energy  extractable with the help of unitary maps assumes the following form: $ \overline{\mathcal{E}}_{\mathrm{U}}=\Tr\left(\rho_BH_B\right)- \frac{\Tr(H_B)}{d_B}$.

Now, we are left with calculating the average energy extractable using CPTP and general quantum maps. The average energy for these maps is given by the equation $\overline{\mathcal{E}}_\mathrm{X}=\overline{E}_\mathrm{X}-\overline{E'}_{\mathrm{X}}$, 
where $\mathrm{X}=\mathrm{CPTP}$, $\mathrm{G}$ and $\overline{E'}_{\mathrm{X}}\coloneqq \Tr\left[\Tr_A\left(U_{BA} \rho_{BA} U^{\dagger}_{BA}\right)H_B\right]$. The term $\overline{E}_\mathrm{X}=\Tr(\rho_B H_B)$ as before. To proceed with the calculations, we first consider the case where we only have those auxiliary systems with a fixed dimension $d_A$. The average extractable energy for this case is given by $\left(\overline{\mathcal{E}}_\mathrm{X}\right)_{d_A}=\Tr\left(\rho_BH_B\right)-\int \int dU_{BA} d\rho_A\Tr\left(\rho'_BH_B\right)$. Evaluating the second term, we get
\begin{align}
   \overline{E'}_{\mathrm{X}}&=\int \int dU_{BA} d\rho_A \Tr\left(\rho'_BH_B\right)=\Tr\biggl[\Tr_A\left(\int dU_{BA} U_{BA} \rho_{BA}U^{\dagger} _{BA}\right)H_B\biggr],\nonumber\\   &=\Tr\biggl[\Tr_A\left(\Tr(\rho_{BA})\frac{\mathbb{I}_{d_B d_A}}{d_B d_A}\right)H_B\biggr]=\frac{\Tr(H_B)}{d_B}. 
\end{align}
Thus, the average energy extractable using CPTP maps and general quantum maps for a fixed auxiliary dimension have the following form: $\left(\overline{\mathcal{E}}_\mathrm{X}\right)_{d_A}= \Tr\left(\rho_BH_B\right)-\frac{\Tr(H_B)}{d_B}$. As this is true for \emph{any} value of $d_A$, we must have $\overline{\mathcal{E}}_\mathrm{X}=\left(\overline{\mathcal{E}}_\mathrm{X}\right)_{d_A}=\Tr\left(\rho_BH_B\right)-\frac{\Tr(H_B)}{d_B}$, where $\overline{\mathcal{E}}_\mathrm{X}$ represents the average extractable energy using processes based on CPTP maps and general quantum maps (without the constraint of having a fixed dimensional auxiliary system). We therefore have $\overline{\mathcal{E}}_{\mathrm{U}}=\overline{\mathcal{E}}_{\mathrm{CPTP}}=\overline{\mathcal{E}}_{\mathrm{G}}$, for all three kinds of quantum processes.
This completes the proof.
\subsection{Calculations of Observation~\ref{obs1}} 
\label{appObs1}
The fluctuation in the extractable energy during the process of extraction of energy from a quantum battery, initialized in a state $\rho_B$ and with a Hamiltonian $H_B$, using unitary maps, is given in Eq.~\eqref{Efluc1}. Using the notation $E$ and $E'$ as before, Eq.~\eqref{Efluc1} can be simplified and conveniently recast in the following way. We can write $\overline{\mathcal{E}^2}_\mathrm{U}=\overline{(E_\mathrm{U}-E'_\mathrm{U})^2}, (\overline{\mathcal{E}}_\mathrm{U})^2=\left(\overline{(E_\mathrm{U}-E'_\mathrm{U})}\right)^2$. Therefore, $(\Delta \mathcal{E}_\mathrm{U})^2=\overline{\mathcal{E}^2}_\mathrm{U}-(\overline{\mathcal{E}}_\mathrm{U})^2=E_\mathrm{U}^2+\overline{E'^2}_\mathrm{U}-2E_\mathrm{U}\overline{E'}_\mathrm{U}-E_\mathrm{U}^2-(\overline{E'}_\mathrm{U})^2+2E_\mathrm{U}\overline{E'}_\mathrm{U}=\overline{E'^2}_\mathrm{U}-(\overline{E'}_\mathrm{U})^2$. The first term $\overline{E'^2}_\mathrm{U}$ can be computed as,
\begin{align*}
    \overline{E'^2}_{\mathrm{U}}&=\int  dU_B \left[\Tr\left(\rho'_BH_B\right)\right]^2=\int dU_B \left[\Tr\left(\rho'_BH_B \otimes \rho'_BH_B\right)\right]=\Tr\left[\left(\int dU_B \rho'^{\otimes2}_B\right)H^{\otimes2}_B\right]=\Tr\left[zH^{\otimes2}_B\right]. 
\end{align*}
The factor $z$ can be worked out in a straightforward manner,
\begin{align*}
    z& \coloneqq\int dU_B  \rho'^{\otimes2}_B=\int dU_B U^{\otimes2}_B\rho^{\otimes2}_B (U^{\dagger} _B)^{\otimes2}=\frac{1}{d_B^2-1}\Biggl[\biggl(\Tr(\rho^{\otimes2}_B)-\frac{\Tr(S\rho^{\otimes2}_B)}{d_B}\biggr)\mathbb{I}^{\otimes 2}_{d_B}+\biggl(\Tr(S\rho^{\otimes2}_B)-\frac{\Tr(\rho^{\otimes2}_B)}{d_B}\biggr)S\Biggr]. 
\end{align*}

Here $S$ refers to the SWAP operator, introduced in Appendix~\ref{appA}. Now $\Tr(\rho^{\otimes2}_B)=(\Tr(\rho_B))^2=1$. Using Eq.~\eqref{A3}, 
$\Tr(S\rho^{\otimes2}_B)=\Tr(\rho^2_B)=\alpha_1$. So we have $z=\frac{1}{d_B^2-1}\left[\biggl(1-\frac{\alpha_1}{d_B}\biggr)\mathbb{I}^{\otimes 2}_{d_B}+\biggl(\alpha_1-\frac{1}{d_B}\biggr)S\right]$.
The second term $(\overline{E'}_\mathrm{U})^2$ is obtained by squaring $\overline{E'}_{\mathrm{U}}$ derived in the Appendix~\ref{B1}. We then have the following formula for the fluctuation in extractable energy,  
\begin{align*}
   (\Delta \mathcal{E}_\mathrm{U})^2=\frac{1}{d_B^2-1}\biggl[\left(1-\frac{\alpha_1}{d_B}\right)\left(\Tr H_B\right)^2+\left(\alpha_1-\frac{1}{d_B}\right)\Tr \left(H_B^2\right)\biggr]-\frac{\left(\Tr H_B\right)^2}{d_B^2}. 
\end{align*}
Using the expansion of the basis of the Gell-Mann matrices of $H_B$ in subsec~\ref{VA}, the values of $\Tr(H_B)$ and $\Tr\left(H^2_B\right)$ can be immediately worked out to obtain the following: $\Tr(H_B)=a_0 d_B,  \Tr\left(H^2_B\right)=d_B\left(a^2_0+\sum_{i=0}^{d^2_B-1} a^2_i\right)$. Using these, we are led to,
\begin{align}
   (\Delta \mathcal{E}_\mathrm{U})^2&=\frac{1}{d_B^2-1}\left[a^2_0\left(d^2_B-1\right)+\left(\sum_{i=0}^{d^2_B-1} a^2_i\right)\left(\alpha_1d_B-1\right)\right]-a^2_0=\frac{d_B \alpha_1 -1}{d_B^2-1} \sum_{i=1}^{d_B^2-1} a_i^2.
\end{align}

This gives the desired expression.
\vspace{-5mm}
\subsection{Calculations of Observation~\ref{obs2}}
\label{appobs2}
Eq.~\eqref{Efluc2} defines the fluctuations in extractable energy appearing in the process of energy extraction from a quantum battery, initialized in a state $\rho_B$ and having a Hamiltonian $H_B$, by employing CPTP maps, or any map from the set of all possible quantum maps, including NCPTP maps. We can simplify Eq.~\eqref{Efluc2} and write $\Delta \mathcal{E}_\mathrm{X}^2=\overline{E'^2}_\mathrm{X}-\overline{E'}_\mathrm{X}^2$, just as we have done while proving Observation~\ref{obs1} in Appendix~\ref{appObs1}. As before, we use the subscript X  to refer to CPTP and G. We restrict ourselves to a fixed dimensional auxiliary $d_A$. With this restriction, let us first consider the case of CPTP maps. Proceeding in the similar way as in Appendix~\ref{appObs1},
\begin{align*}
      \overline{E'^2}_{\mathrm{CPTP}}&=\int \int  dU_{BA} d\rho_A \left[\Tr\left(\rho'_BH_B\right)\right]^2=\Tr\left[zH^{\otimes2}_B\right].  
\end{align*}
Working out the factor $z$,
\begin{align*}
    z& \coloneqq\int \int  dU_{BA} d\rho_A  \rho'^{\otimes2}_B=\int \int  dU_{BA} d\rho_A \biggl[ \Tr_A\left(U_{BA} \rho_{BA} U^{\dagger}_{BA}\right) \otimes \Tr_A\left(U_{BA} \rho_{BA} U^{\dagger}_{BA}\right)\biggr]. 
\end{align*}
 At this point, we use the following identity (see Appendix \ref{appC} for the proof),
\begin{align}
     \Tr_A\left(U_{BA} \rho_{BA} U^{\dagger}_{BA}\right)\otimes \Tr_A\left(U_{BA} \rho_{BA} U^{\dagger}_{BA}\right)=\Tr_A\left(U_{BA} \rho_{BA} U^{\dagger}_{BA} \otimes U_{BA} \rho_{BA} U^{\dagger}_{BA}\right). \label{id}
\end{align}
Using Eq.~\eqref{id}, we can evaluate $z$, 
\begin{align*}
    z &=\Tr_A\biggl[\int \int  dU_{BA} d\rho_A \biggl(U_{BA} \rho_{BA} U^{\dagger}_{BA} \otimes U_{BA} \rho_{BA} U^{\dagger}_{BA}\biggr)\biggr],  \nonumber\\
    &=\frac{1}{d_B^2d_A^2-1}\Tr_A \left[\left(\Tr Y-\frac{\Tr\left(S_1Y\right)}{d_Bd_A}\right)\mathbb{I}^{\otimes 2}_{d_Bd_A}+\left(\Tr\left(S_1Y\right)-\frac{\Tr Y}{d_Bd_A}\right)S_1\right]. 
\end{align*}
Where $S_1$ is introduced in Appendix~\ref{appA} and $Y\coloneqq\int d\rho_A \rho_{BA}^{\otimes2}$. 
Now, $\Tr(Y)=\int d\rho_A \Tr\left(\rho_{BA}^{\otimes2}\right)=\int d\rho_A \left[\Tr\left(\rho_{BA}\right)\right]^2=1$.
Labelling $\alpha=\Tr(S_1Y)$, we have $\alpha=\int d\rho_A \Tr\left(S_1\rho_{BA}^{\otimes2}\right)=\int d\rho_A \Tr\left(\rho_{BA}^2\right)$.
For CPTP maps, $\rho_{BA}=\rho_B \otimes \rho_A$. We then have,
\begin{align}
   \alpha &=\int d\rho_A \Tr\biggl[\left(\rho_B \otimes \rho_A\right)\left(\rho_B \otimes \rho_A\right)\biggr]=\int d\rho_A \Tr\left(\rho^2_B \otimes \rho^2_A\right)=\Tr\left(\rho_B^2\right)\int \Tr\left(\rho_A^2\right) d\rho_A. \label{x3}
\end{align}
Ref.~\cite{mixedness-avg} explicitly calculated the integral appearing in Eq.~\eqref{x3} by considering the purification of the system $A$ (the auxiliary system in this work) over a bipartite Hilbert space, consisting of the concerned system along with an ancilla system $C$, with a dimension $d_C$. The result is given in the following form,
\begin{equation}
    \int \Tr\left(\rho_A^2\right) d\rho_A=\frac{\Gamma(d_Ad_C)}{\Gamma(d_Ad_C+2)} \Tr J(2). \label{x}
\end{equation}
Where $J$ is a $d_A\times d_A$ matrix, with the elements given by,
\begin{equation}
\begin{split}
    J_{ij}(r)=\sum_{p=0}^{d_A-1} \frac{\Gamma(d_C-d_A+r+p+1)\Gamma(j+1)\Gamma(r+1)^2}{\Gamma(d_C-d_A+i+1)\Gamma(i-p+1)\Gamma(r+p-i+1)\Gamma(j-p+1)\Gamma(r-j+p+1)\Gamma(p+1)}. \label{Xij}
\end{split}
\end{equation}
Here $\Gamma(n)=\int_0^{\infty} e^{-x} x^{n-1} dx$ is the Gamma function. 
\vspace{-0.12em}
Using the properties of Gamma function, Eq.~\eqref{x} can be simplified to give, 
\begin{equation}
    \int \Tr\left(\rho_A^2\right) d\rho_A=\frac{d_A+d_C}{d_Ad_C+1}. \label{xf}
\end{equation}
The derivation of Eq.~\eqref{xf} is given in Appendix \ref{appC}. 
In our work, we assume $d_C=d_A$. 
Putting Eq.~\eqref{xf} in Eq.~\eqref{x3} with this assumption and using $\alpha_1=\Tr\left(\rho_B^2\right)$, we get 
\begin{equation}
\alpha=\frac{2d_A}{d_A^2+1}\alpha_1. \label{alpha}   
\end{equation}
The expression for $z$ can be written as,
\begin{align}
    z &=\frac{1}{d_B^2d_A^2-1}\Tr_A \left[\left(1-\frac{\alpha}{d_Bd_A}\right)\mathbb{I}^{\otimes 2}_{d_Bd_A}+\left(\alpha-\frac{1}{d_Bd_A}\right)S_1\right]=\frac{1}{d_B^2d_A^2-1} \left[\left(1-\frac{\alpha}{d_Bd_A}\right)d^2_A\mathbb{I}^{\otimes 2}_{d_B}+\left(\alpha-\frac{1}{d_Bd_A}\right)d_AS\right]. \label{z11}
\end{align}
This gives,
\begin{align*}
\overline{E'^2}_{\mathrm{CPTP}}&=\frac{1}{d_B^2d_A^2-1}\Tr\biggl[\left(1-\frac{\alpha}{d_Bd_A}\right)d^2_A\mathbb{I}^{\otimes 2}_{d_B} H_B^{\otimes2}+\left(\alpha-\frac{1}{d_Bd_A}\right)d_ASH_B^{\otimes2}\biggr],\nonumber\\
&=\frac{1}{d_A^2d_B^2-1}\biggl[\left(1-\frac{\alpha}{d_Ad_B}\right)d_A^2\left(\Tr H_B\right)^2+\left(\alpha-\frac{1}{d_Ad_B}\right)d_A\Tr \left(H_B^2\right)\biggr]. \label{Ecp2}
\end{align*}
Using the expressions for $\Tr \left(H_B^2\right)$ and $\left(\Tr H_B\right)^2$ in terms of the basis expansion given in subsec.~\ref{VA}, we can now write down the expression for fluctuation in extractable energy corresponding to CPTP maps for a fixed battery and auxiliary dimension in a compact form,
\begin{align}
    (\Delta \mathcal{E}_{\mathrm{CPTP}})^2_{d_A}&= \frac{1}{d_A^2d_B^2-1}\Biggl[\left(1-\frac{\alpha}{d_Ad_B}\right)d_A^2\left(\Tr H_B\right)^2+\left(\alpha-\frac{1}{d_Ad_B}\right)d_A\Tr \left(H_B^2\right)\Biggr]-\frac{\left(\Tr H_B\right)^2}{d_B^2}= \frac{d_Bd_A\alpha-1}{d_B^2d_A^2-1} \sum_{i=1}^{d_B^2-1} a_i^2.
\end{align}
We are left with the case where a map is chosen from the set of general quantum maps, including NCPTP maps. In this case, the initial state of the composite battery-auxiliary system $\rho_{BA}$ is no longer constrained to be of the product form $\rho_B \otimes \rho_A$ and can be any valid quantum state, including entangled state. 
Here we restrict ourselves to the case where $\rho_{BA}$ is a pure entangled state. Accordingly, $ \alpha =\int d\rho_A \Tr_{BA}\left(\rho_{BA}^2\right)=1$. Then, 
\begin{align*}
     z&=\frac{1}{d_B^2d_A^2-1}\Tr_A \left[\left(1-\frac{1}{d_Bd_A}\right)\left(\mathbb{I}^{\otimes 2}_{d_Bd_A}+S_1\right)\right]=\frac{1}{d_B^2d_A^2-1}\left(1-\frac{1}{d_Bd_A}\right)\left(d^2_A\mathbb{I}^{\otimes 2}_{d_B}+d_AS\right),
\end{align*}
where we have used Eq.~\eqref{swap1}.
The fluctuation in extractable energy in this case becomes,
\begin{align}
    (\Delta \mathcal{E}_\mathrm{G})^2_{d_A}&=\frac{1}{d_A^2d_B^2-1}\biggl[\left(1-\frac{1}{d_Ad_B}\right)d_A^2\left(\Tr H_B\right)^2+\left(1-\frac{1}{d_Ad_B}\right)d_A\Tr \left(H_B^2\right)\biggr]-\frac{\left(\Tr H_B\right)^2}{d_B^2}=\frac{1}{d_Bd_A+1} \sum_{i=1}^{d_B^2-1} a_i^2.
\end{align}
Thus we obtain the desired expressions. 
\vspace{-0.7cm}
\subsection{Calculations of Theorem~\ref{th3}}\label{appth3}
The fluctuation in extractable energy arising from general quantum maps, averaged over auxiliary dimensions from $d_A=2$ up to a limit $d_A=n$ can be written as $\left\langle(\Delta \mathcal{E}_{\mathrm{G}})^2\right\rangle_{n}= \frac{1}{n}\sum_{d_A=2}^{n+1}\frac{\beta}{d_Bd_A+1} =\frac{\beta}{nd_B}\sum_{d_A=2}^{n+1}\frac{1}{d_A+\frac{1}{d_B}}$, where we have denoted $\beta=\sum_{i=1}^{d_B^2-1} a_i^2$. Using the following chain of inequalities $\frac{1}{d_A+1}<\frac{1}{d_A+\frac{1}{d_B}}<\frac{1}{d_A}$, we can bound the summation as
$\sum_{d_A=2}^{n+1}\frac{1}{d_A+1}<\sum_{d_A=2}^{n+1}\frac{1}{d_A+\frac{1}{d_B}}<\sum_{d_A=2}^{n+1}\frac{1}{d_A}$. Both upper and lower bounds correspond to a finite sum of the Harmonic sequence $\left\{\frac{1}{m}\right\}$ which is approximated by the following finite sum approximation $\sum_{m=1}^N \frac{1}{m} \approx \ln N+\gamma$, with $\gamma\approx0.577$. This enables us to calculate the bounds exactly. Consequently, the lower bound is given by $\sum_{d_A=2}^{n+1}\frac{1}{d_A+1}=\sum_{d_A=1}^{n+2}\frac{1}{d_A}-1-\frac{1}{2}=\ln (n+2)+\gamma-\frac{3}{2}$. 
Similarly, the upper bound is calculated as $\sum_{d_A=2}^{n+1}\frac{1}{d_A}=\sum_{d_A=1}^{n+1}\frac{1}{d_A}-1=\ln (n+1)+\gamma-1$. Then, the finite-term average fluctuation in extractable energy generated by general quantum maps is bounded as $ \frac{1}{n}\left[\ln (n+2)+\gamma-\frac{3}{2}\right]<\sum_{d_A=2}^{n+1}\frac{1}{d_A+\frac{1}{d_B}}<\frac{1}{n}\left[\ln (n+1)+\gamma-1\right]$. 
Considering large $n$, $\frac{\ln n}{n}$ dominates over $\frac{\gamma-1}{n}$ and hence determines the scaling
\begin{equation}
   \left\langle(\Delta \mathcal{E}_{\mathrm{G}})^2\right\rangle_{n}\sim \frac{\ln n}{n}. 
   \label{scaling1}
\end{equation}
Likewise, we can write the expression for the fluctuation in extractable energy arising from CPTP maps, averaged over auxiliary dimensions from $d_A=2$ to $d_A=n$ as $\left\langle(\Delta \mathcal{E}_{\mathrm{CPTP}})^2\right\rangle_{n}= \frac{\beta}{n}\sum_{d_A=2}^{n+1}\frac{d_Bd_A\alpha-1}{d_B^2d_A^2-1} =\frac{\beta}{nd_B}\sum_{d_A=2}^{n+1}\biggl[\frac{d_Bd_A\alpha}{d_B^2d_A^2-1}-\frac{1}{d_B^2d_A^2-1}\biggr]$.
Using the form of $\alpha$ given in Eq.~\eqref{alpha}, we have
\begin{align*}
   \left\langle(\Delta \mathcal{E}_{\mathrm{CPTP}})^2\right\rangle_{n}&= \frac{\beta}{n}\sum_{d_A=2}^{n+1}\left[\frac{2\alpha_1d^2_A}{\left(d^2_A+1\right)d_B\left(d_A^2-\frac{1}{d_B^2}\right)}-\frac{1}{d_B^2\left(d_A^2-\frac{1}{d^2_B}\right)}\right],\nonumber\\
   &=\frac{\beta}{n}\sum_{d_A=2}^{n+1}\left[\frac{1}{\left(d_A^2-\frac{1}{d^2_B}\right)}\left\{\frac{2\alpha_1}{d_B}\left[1-\frac{1}{\left(1+\frac{1}{d^2_B}\right)}\right]-\frac{1}{d^2_B}\right\}+\frac{1}{\left(1+\frac{1}{d^2_B}\right)}\frac{1}{\left(d^2_A+1\right)}\right]. 
\end{align*}
Introducing two functions $ p(d_B,\alpha_1)=\frac{2\alpha_1}{d_B}\left[1-\frac{1}{\left(1+\frac{1}{d^2_B}\right)}\right]-\frac{1}{d^2_B}$ and
      $q(d_B)=\frac{1}{1+\frac{1}{d^2_B}}$, we can write the expression for $ \left\langle(\Delta \mathcal{E}_{\mathrm{CPTP}})^2\right\rangle_{n}$ in a compact form as $\left\langle(\Delta \mathcal{E}_{\mathrm{CPTP}})^2\right\rangle_{n}=\frac{\beta}{n}\sum_{d_A=2}^{n+1}\left[\frac{p}{d_A^2-\frac{1}{d^2_B}}-\frac{q}{d^2_A+1}\right]$. We now use the following chain of inequalities $\frac{1}{d^2_A+d_A}<\frac{1}{d_A^2-\frac{1}{d^2_B}},\frac{1}{d^2_A+1}<\frac{1}{d^2_A-d_A}$ to bound both terms of the above sum as $ \sum_{d_A=2}^{n+1}\frac{1}{d^2_A+d_A}<\sum_{d_A=2}^{n+1}\frac{1}{d_A^2-\frac{1}{d^2_B}},\sum_{d_A=2}^{n+1}\frac{1}{d^2_A+1}<\sum_{d_A=2}^{n+1}\frac{1}{d^2_A-d_A}$. 
The bounds can be computed by expressing them as a difference between sums of two Harmonic sequences and using the finite sum approximation as before. The lower bound is given as $\sum_{d_A=2}^{n+1}\frac{1}{d_A^2+d_A}=\sum_{d_A=2}^{n+1}\left(\frac{1}{d_A}-\frac{1}{d_A+1}\right)=\frac{1}{2}-\frac{1}{n+2}$, while the upper bound is $\sum_{d_A=2}^{n+1}\frac{1}{d_A^2-d_A}=\sum_{d_A=2}^{n+1}\left(\frac{1}{d_A-1}-\frac{1}{d_A}\right)=1-\frac{1}{n+1}$. The finite-term average fluctuation in extractable energy generated by CPTP maps is then bounded in the following manner: $\frac{1}{2n}-\frac{1}{n\left(n+2\right)}< \frac{1}{n}\sum_{d_A=2}^{n+1}\frac{1}{d_A^2+1},\frac{1}{n}\sum_{d_A=2}^{n+1}\frac{1}{d^2_A-\frac{1}{d^2_B}}<\frac{1}{n}-\frac{1}{n\left(n+1\right)}$.
For large values of $n$, $\frac{1}{n(n+1)}$ approaches to zero very quickly as compared to $\frac{1}{n}$ and hence can be neglected. Thus the scaling is given by
\begin{equation}  \left\langle(\Delta\mathcal{E}_{\mathrm{CPTP}})^2\right\rangle_{n} \sim \frac{1}{n}. \label{scaling2}
\end{equation}
This completes the proof. 
\vspace{-10mm}
\section{Some derivations} \label{appC}
This section contains detailed derivations of Eqs.~\eqref{id} and~\eqref{xf}. Also, we give the calculations for the intermediate steps in the derivations of conditions~\eqref{ucptpmain} and~\eqref{uncptpmain}.

We will first prove Eq.~\eqref{id}. Let $\left\{P^B_i\right\}$ and $\left\{Q^A_i\right\}$ be the basis elements of the operator spaces $\mathcal{L}(\mathcal{H}_B)$ and $\mathcal{L}(\mathcal{H}_A)$ respectively. Then, the set $\left\{P^B_i \otimes Q^A_j\right\}$ constitutes a basis of the tensor product space $\mathcal{L}(\mathcal{H}_B) \otimes \mathcal{L}(\mathcal{H}_A)$. For finite-dimensional vector spaces, $\mathcal{L}(\mathcal{H}_B) \otimes \mathcal{L}(\mathcal{H}_A) \cong \mathcal{L}(\mathcal{H}_B \otimes\mathcal{H}_A)$, so that the set $\left\{P^B_i \otimes Q^A_j\right\}$ is also a basis for the space $\mathcal{L}(\mathcal{H}_B \otimes\mathcal{H}_A)$. We can expand the operators $U_{BA}$, $\rho_{BA}$, $U^\dagger_{BA}$ $\in \mathcal{L}(\mathcal{H}_B \otimes\mathcal{H}_A)$ in this basis as $U_{BA}=\sum_{i,j} u_{ij} P^B_i \otimes Q^A_j,\hspace{0.3cm}
    \rho_{BA}=\sum_{i,j} \rho_{ij} P^B_i \otimes Q^A_j,\hspace{0.3cm}
    U^\dagger_{BA}=\sum_{i,j} \Tilde{u}_{ij} P^B_i \otimes Q^A_j$
($\Tilde{u}_{ij} \neq u_{ij}^{*}$ in general, where $u_{ij}^{*}$ denotes complex conjugate of $u_{ij}$). Using this, we can write,
\begin{align*}
    \Tr_A\left(U_{BA}\rho_{BA}U^\dagger_{BA}\right)&=\Tr_A\Biggl[\sum_{ijklmn}u_{ij}\rho_{kl}\Tilde{u}_{mn}\biggl(P^B_i \otimes Q^A_j\biggr)\biggl(P^B_k \otimes Q^A_l\biggr)\biggl(P^B_m \otimes Q^A_n\biggr)\Biggr],\\
    &=\sum_{ijklmn}u_{ij}\rho_{kl}\Tilde{u}_{mn}\Tr\biggl(Q^A_jQ^A_lQ^A_n\biggr)P^B_iP^B_k P^B_m.
\end{align*}
Here we have used the linearity property of the partial trace. Thus we have
\begin{align*}  &\Tr_A\left(U_{BA}\rho_{BA}U^\dagger_{BA}\right)\otimes \Tr_A\left(U_{BA}\rho_{BA}U^\dagger_{BA}\right)\\&=\sum_{ijklmn}\sum_{opqrst}u_{ij}\rho_{kl}\Tilde{u}_{mn}u_{op}\rho_{qr}\Tilde{u}_{st}\Tr\biggl(Q^A_jQ^A_lQ^A_n\biggr)\Tr\biggl(Q^A_pQ^A_rQ^A_t\biggr)\biggl(P^B_iP^B_kP^B_m \otimes P^B_oP^B_qP^B_s\biggr). 
\end{align*}
Now, $U_{BA} \rho_{BA} U^{\dagger}_{BA} \otimes U_{BA} \rho_{BA} U^{\dagger}_{BA} \equiv U_{B_1A_1} \rho_{B_1A_1} U^{\dagger}_{B_1A_1} \otimes U_{B_2A_2} \rho_{B_2A_2} U^{\dagger}_{B_2A_2} \in \mathcal{L}\left( \left(\mathcal{H}_{B_1}\otimes\mathcal{H}_{A_1}\right)\otimes \left(\mathcal{H}_{B_2}\otimes\mathcal{H}_{A_2}\right)\right)$. We have used the dummy indices $A_1,A_2$ and $B_1,B_2$ to refer to the two copies of systems $A$ and $B$ respectively for the sake of clarity and ease of calculation. Now we have
\begin{align*}
&\Tr_A \left( U_{BA} \rho_{BA} U_{BA}^\dagger \otimes U_{BA} \rho_{BA} U_{BA}^\dagger \right) 
\\&= \Tr_{A_1 A_2} \left( U_{B_1 A_1} \rho_{B_1 A_1} U_{B_1 A_1}^\dagger \otimes U_{B_2 A_2} \rho_{B_2 A_2} U_{B_2 A_2}^\dagger \right) \\
&= \Tr_{A_1 A_2} 
\sum_{ijklmn} \sum_{opqrst}\left[  u_{ij} \rho_{kl} \tilde{u}_{mn} u_{op} \rho_{qr} \tilde{u}_{st}
\left( P_i^{B_1} P_k^{B_1} P_m^{B_1} \otimes Q_j^{A_1} Q_l^{A_1} Q_n^{A_1} \right)  \otimes \left( P_o^{B_2} P_q^{B_2} P_s^{B_2} \otimes Q_p^{A_2} Q_r^{A_2} Q_t^{A_2} \right) 
\right] \\
&= \sum_{ijklmn} \sum_{opqrst} u_{ij} \rho_{kl} \tilde{u}_{mn} u_{op} \rho_{qr} \tilde{u}_{st} 
\left[ \Tr \left( Q_j^{A_1} Q_l^{A_1} Q_n^{A_1} \right) \Tr \left( Q_p^{A_2} Q_r^{A_2} Q_t^{A_2} \right)  \times \left( P_i^{B_1} P_k^{B_1} P_m^{B_1} \otimes P_o^{B_2} P_q^{B_2} P_s^{B_2} \right) \right] \\
&= \Tr_A \left( U_{BA} \rho_{BA} U_{BA}^\dagger \right) \otimes \Tr_A \left( U_{BA} \rho_{BA} U_{BA}^\dagger \right).
\end{align*}
This completes the proof.

Next, we give the proof of Eq.~\eqref{xf}. Eq.~\eqref{Xij} gives us the expression for the matrix $J(2)$. Putting $r=2$ and taking the trace, we have
\begin{equation}
    \Tr J(2)=\sum_{i=0}^{d_A-1}\sum_{p=0}^{d_A-1} \frac{\Gamma(d_C-d_A+p+3)\Gamma(i+1)\left(\Gamma(3)\right)^2}{\Gamma(d_C-d_A+i+1)\left(\Gamma(i-p+1)\right)^2\left(\Gamma(3-i+p)\right)^2\Gamma(p+1)} . \label{gamma}
\end{equation}
We note that the arguments of the Gamma functions appearing in Eq.~\eqref{gamma} are all integers and therefore must be strictly positive since $\Gamma(x)$ for integer $x$ is defined only for $x>0$.  Ref.~\cite{mixedness-avg} assumed that $d_C\geq d_A$ so that $d_C-d_A \geq0$, leading to $d_C-d_A+i+1>0$. Since $p$ starts from zero, $p+1>0$. Thus, we must choose $i,p$ such that $i-p+1>0$ and $3-i+p>0$. This implies $i-p=0,1,2$. Note that for $d_A=2$, $i$ and $p$ can take values $0,1$ only, which means the case $i-p=2$ does not occur. We will handle the case $d_A=2$ separately. First, let us compute for $d_A\geq3$ for which $i-p=2$ is valid. 

Using the relation between $i$ and $p$, we can calculate $\Tr J(2)$ by writing $i$ in terms of $p$ or vice versa. We choose the former approach here. For the case $i=p$, $p$ runs from $0$ to $d_A-1$ as usual. However, for $i=p+1$, p can only take values up to $d_A-2$, starting from $0$. This is because $p=d_A-1$ would mean $i=d_A$, which is not possible since $i$ can go only up to $d_A-1$, as seen in Eq.~\eqref{gamma}. Similarly, the sum on $p$ for $i=p+2$ can go up to $d_A-3$. Keeping this in mind and substituting $i=p,p+1,p+2$ in Eq.~\eqref{gamma}, we have
\begin{align*}
    \Tr J(2)&=\sum_{p=0}^{d_A-1} \frac{\Gamma(d_C-d_A+p+3)\Gamma(p+1)\Gamma(3)^2}{\Gamma(d_C-d_A+p+1)\Gamma(1)\Gamma(p+1)\Gamma(3)^2}+\sum_{p=0}^{d_A-2} \frac{\Gamma(d_C-d_A+p+3)\Gamma(p+2)\Gamma(3)^2}{\Gamma(d_C-d_A+p+2)\Gamma(p+1)\Gamma(2)^4}\\&+\sum_{p=0}^{d_A-3} \frac{\Gamma(d_C-d_A+p+3)\Gamma(p+3)\Gamma(3)^2}{\Gamma(d_C-d_A+p+3)\Gamma(1)^2\Gamma(p+1)\Gamma(3)^2}.
\end{align*}
As $\Gamma(x)=(x-1)!$ for any integer $x>0$, this can be simplified to
\begin{align*}
    \Tr J(2)&= \sum_{p=0}^{d_A-1}\left(d_C-d_A+p+1\right)\left(d_C-d_A+p+2\right) +   \sum_{p=0}^{d_A-2}4(p+1)\left(d_C-d_A+p+2\right)+\sum_{p=0}^{d_A-3} (p+1)(p+2),\\&=J_1+J_2+J_3.
\end{align*}

Using the formulae $\sum_{i=1}^N p=\frac{N(N+1)}{2}$ and $\sum_{i=1}^N p^2=\frac{N(N+1)(2N+1)}{6}$, the terms $J_1$, $J_2$ and $J_3$ can be evaluated as
\begin{align*}
  J_1&=\sum_{p=0}^{d_A-1}\left(d_C-d_A+p+1\right)\left(d_C-d_A+p+2\right)=d_A\left(d_C-d_A\right)^2+\left(d_C-d_A\right)\left(d_A^2+2d_A\right)
  +\frac{2d_A^3+6d_A^2+4d_A}{6}.
\end{align*}
\begin{align*}
    J_2&=\sum_{p=0}^{d_A-2}4(p+1)\left(d_C-d_A+p+2\right),\\
   &=4\left(d_C-d_A\right)\left[\frac{\left(d_A-2\right)\left(d_A-1\right)}{2}+\left(d_A-1\right)\right]+4\biggl[\frac{\left(d_A-2\right)\left(d_A-1\right)\left(2d_A-3\right)}{6}+\frac{3\left(d_A-2\right)\left(d_A-1\right)}{2}+2\left(d_A-1\right) \biggr].
\end{align*}
\begin{align*}
    J_3&=\sum_{p=0}^{d_A-3} (p+1)(p+2)=\frac{\left(d_A-3\right)\left(d_A-2\right)(2d_A-5)}{6}+\frac{3\left(d_A-3\right)\left(d_A-2\right)}{2}+2\left(d_A-2\right). 
\end{align*}
Putting the expressions for $J_1$, $J_2$ and $J_3$ and simplifying, we get
\begin{equation}
    \Tr J(2)=\left(d_A+d_C\right)d_Ad_C. \label{Tr_X(2)}
\end{equation}
We now do the calculations for $d_A=2$. In this case, the allowed values of $i$ are $p,p+1$ which gives us  $\Tr J(2)=J_1+J_2$; the term $J_3$ corresponds to the case $i=p+2$ and hence does not occur. Putting $d_A=2$, $ J_1$ and $J_2$ are obtained as $\sum_{p=0}^{1}\left(d_C-d_A+p+1\right)\left(d_C-d_A+p+2\right)=2d_C^2$ and $4\left(d_C-d_A+2\right)=4d_C$ respectively. Then $ \Tr J(2)=2d_C^2+2^2d_C=\left(d_A+d_C\right)d_Ad_C$. 
We have thus verified the relation $\Tr J(2)=\left(d_A+d_C\right)d_Ad_C$ for all possible values of $d_A$. Now using Eq.~\eqref{Tr_X(2)} in Eq.~\eqref{x}, 
\begin{align}
 \int \Tr\left(\rho_A^2\right) d\rho_A&=\frac{\left(d_Ad_C-1\right)!}{\left(d_Ad_C+1\right)!}\left(d_A+d_C\right)d_Ad_C=\frac{d_A+d_C}{d_Ad_C+1}. 
 \label{avgpurity}
\end{align}
Next, we provide the steps leading to Eqs.~\eqref{ucptpmain} and~\eqref{uncptpmain}. We have,
\begin{align*}
    \frac{(\Delta\mathcal{E}_{\mathrm{U}})^2-(\Delta\mathcal{E}_{\mathrm{CPTP}})^2}{\sum_{i=1}^{d^2_B-1}a^2_i}&=\frac{d_B\alpha_1-1}{d^2_B-1}-\frac{d_Ad_B\alpha-1}{d^2_Ad^2_B-1}
    \approx \frac{1}{d_Bd_A^2}\left[d_A^2\alpha_1-\frac{d_A^2}{d_B}-d_A\alpha+\frac{1}{d_B}\right]
    \approx \frac{\left(d_A^2\alpha_1-d_A\alpha\right)}{d_Bd_A^2}. 
\end{align*}
\begin{align*}
    \frac{(\Delta\mathcal{E}_{\mathrm{G}})^2-(\Delta\mathcal{E}_{\mathrm{U}})^2}{\sum_{i=1}^{d^2_B-1}a^2_i}&=\frac{1}{d_Bd_A+1}-\frac{d_B\alpha_1-1}{d^2_B-1}
    \approx \frac{1}{d_B\left(d_A+\frac{1}{d_B}\right)}\left[1-d_A\alpha_1+\frac{d_A-\alpha_1}{d_B}\right]=\frac{1}{d_B}\left[\frac{d_A+d_B}{1+d_Ad_B}-\alpha_1\right]. 
\end{align*}

\end{document}